\documentclass[11pt,a4paper,oneside]{amsart}
\usepackage{mathrsfs,amssymb}
\usepackage[a4paper,marginparwidth=2cm,margin=1in]{geometry}

\usepackage[colorlinks,pdfstartview=FitH,
            linkcolor={red!90!black},
            citecolor={green!60!black},
            urlcolor={blue}
            ]{hyperref}
\usepackage{tikz}
\usepackage{slashed}
\usepackage{float}

\numberwithin{equation}{section}
\theoremstyle{plain}   
\newtheorem{theorem}{Theorem}[section]
 \newtheorem{proposition}[theorem]{Proposition}
\newtheorem{lemma}[theorem]{Lemma}
\newtheorem{corollary}[theorem]{Corollary}
\newtheorem{conjecture}[theorem]{Conjecture}

\theoremstyle{definition}
\newtheorem{definition}[theorem]{Definition}

\theoremstyle{remark}
\newtheorem{remark}[theorem]{Remark}

\makeatletter
\@namedef{subjclassname@2020}{\textup{2020} Mathematics Subject Classification}
\makeatother

\begin{document}

\title[]{Tilted spacetime positive mass theorem with arbitrary ends}

\author[]{Daoqiang Liu}

\address{
     School of Mathematical Sciences,
     Capital Normal University, 100048,
     Beijing, China}
\email{\href{mailto:dqliumath@cnu.edu.cn}{dqliumath@cnu.edu.cn}}


\date{November 27, 2023}


\begin{abstract}
In this paper, we prove the spacetime positive mass theorem for asymptotically flat spin initial data sets with arbitrary ends and a non-compact boundary. Moreover, we demonstrate a quantitative shielding theorem, subject to the tilted boundary dominant energy condition. Our results are established by solving a mixed boundary value problem for the Dirac-Witten operator with a Callias potential.
\end{abstract}

\maketitle

\section{Introduction}\label{sec:intro}
The positive mass theorem is one of the most fundamental results in differential geometry and mathematical general relativity, which states that if a complete asymptotically flat manifold $(M,g)$ has non-negative scalar curvature, then the ADM mass of an asymptotically flat end is non-negative and the mass vanishes if and only if $(M,g)$ is isometric to the Euclidean space. This result was first given by Schoen and Yau \cite{SY79a,SY79b,SY81a} (see also \cite{Sch89}) for dimensions $\leq 7$ using minimal hypersurfaces. Furthermore, the non-negativity result was generalized to initial data sets by Schoen-Yau \cite{SY81b} in dimension three via Jang’s equation and by Eichmair–Huang–Lee–Schoen \cite{EHLS16} for dimensions $\leq 7$ using marginally outer trapped hypersurfaces (MOTS), instead of minimal hypersurfaces. 
Witten gave a different proof of the positive mass theorem by spinors under the spin condition \cite{Wit81}, as explained in detail in \cite{Bar86,PT82}. For higher dimension cases, proofs were presented by Schoen-Yau \cite{SY17} and Lohkamp \cite{Loh16,Loh17}.

In \cite{SY88,LUY20}, Schoen-Yau and Lesourd-Unger-Yau conjectured that the positive mass theorem holds for an asymptotically flat manifold with arbitrary ends.
\begin{conjecture}\label{con:SY-LUY}
Let $(M^{n\geq 3},g)$ be an $n$-dimensional complete Riemannian manifold which has non-negative scalar curvature. Let $\mathcal{E}\subseteq M$ be a single asymptotically flat end in $M$. Then the ADM mass of $\mathcal{E}$ is non-negative.
\end{conjecture}
Using $\mu$-bubbles, Lesourd, Unger, and Yau \cite{LUY21} proved Conjecture \ref{con:SY-LUY} for $n\leq 7$, assuming that the chosen end $\mathcal{E}$ is asymptotically Schwarzschild. Later, the extra assumption was relaxed in \cite{LLU22}. On the other hand, Bartnik and Chru\'{s}ciel \cite{BC03,BC05} gave a different proof of Conjecture \ref{con:SY-LUY} via Witten's approach in the spin setting. Moreover, they proved that the ADM mass of the end $\mathcal{E}$ is zero if and only if $(M,g)$ is isometric to the Euclidean space. Recently, Cecchini and Zeidler \cite{CZ21b} addressed Conjecture \ref{con:SY-LUY} by using Callias operators beyond Witten's spinorial method. Subsequently, the cases of asymptotically hyperbolic ends and asymptotically flat initial data ends were respectively extended by Chai-Wan \cite{CW22} and Cecchini-Lesourd-Zeidler \cite{CLZ23} using similar techniques. We also note that a proof of Conjecture \ref{con:SY-LUY} for $n\leq 7$ based on the idea of density theorem and Lohkamp's compactification was found by Zhu \cite{Zhu22}.

In \cite{ABdL16}, Almaraz, Barbosa, and de Lima introduced a variation of the classical positive mass theorem for asymptotically flat manifolds carrying a non-compact boundary with three different proofs via the arguments of Schoen-Yau \cite{SY79a,SY81a}, Witten \cite{Wit81}, and Miao \cite{Mia02} respectively. Another proof of the positive mass theorem of such manifolds was given by Chai \cite{Cha18} using free boundary minimal surfaces. Batista and de Lima \cite{BdL23} recently also provided a proof based on a harmonic function technique in \cite{BKKS22}. The case of asymptotically hyperbolic spin manifolds was presented in \cite{AdL20}. Almaraz, de Lima, and Mari \cite{AdLM21} extended the results in \cite{ABdL16,AdL20} to the setting of spin initial data sets. For the asymptotically flat case, they showed that 
\begin{theorem}[{\cite{AdLM21}}]\label{thm:AdLM21-flat}
Let $(M^{n\geq 3},g,k,\mathcal{E},\Sigma)$ be an $n$-dimensional complete asymptotically flat spin initial data set with a non-compact boundary such that the interior dominant energy condition $\mu-|J|\geq 0$ and the tangential boundary dominant energy condition $H_{\Sigma} \geq |k(\eta,\cdot)^{\top}|$. Then $E_{\mathcal{E}} \geq |\widehat{P}_{\mathcal{E}}|$.
\end{theorem}
We refer to Section \ref{sec:pre} for precise statements of terms used above. Here the notion of completeness of an initial data set with boundary refers to metrically completeness, and from now on we will just use this short version for convenience. Moreover, Almaraz, de Lima, and Mari \cite{AdLM21} proved the rigidity statement of Theorem \ref{thm:AdLM21-flat}, that is, if $E_{\mathcal{E}}=0$ then $(M,g)$ may be isometrically embedded in the Minkowski space with second fundamental form $k$, and $\Sigma$ is totally geodesic (as a hypersurface in $M$) lies on the boundary of Minkowski half-space. 

Recently, Chai \cite{Cha23} proved the following spacetime positive mass theorem.
\begin{theorem}[{\cite{Cha23}}]
Let $(M^{n\geq 3},g,k,\mathcal{E},\Sigma)$ be an $n$-dimensional complete asymptotically flat spin initial data set with a non-compact boundary such that the interior dominant energy condition $\mu-|J|\geq 0$ and the normal boundary dominant energy condition $\theta_{\Sigma}^{\pm} \geq 0$. Then $E_{\mathcal{E}} \pm (P_{\mathcal{E}})_{n} \geq 0$.
\end{theorem}
Chai \cite{Cha23} also generalized the tangential and normal boundary dominant energy condition to the tilted boundary dominant energy condition and proved the non-negativity of a new mass type invariant in general.
\begin{theorem}[{\cite{Cha23}}]
Let $(M^{n\geq 3},g,k,\mathcal{E},\Sigma)$ be an $n$-dimensional complete asymptotically flat spin initial data set with a non-compact boundary such that the interior dominant energy condition $\mu-|J|\geq 0$ and the tilted boundary dominant energy condition $H_{\Sigma} \pm \cos\alpha\, {\rm tr}_{\Sigma}k \geq \sin\alpha\, |k(\eta, \cdot)^{\top}|$ for some $\alpha\in [0,\frac{\pi}{2}]$. Then $E_{\mathcal{E}} \pm \cos\alpha\, (P_{\mathcal{E}})_n \geq \sin \alpha\, |\widehat{P}_{\mathcal{E}}|$.
\end{theorem}
In this paper, we would like to establish some general results of the spacetime positive mass theorem based on the ideas of Chai \cite{Cha23} and Cecchini-Lesourd-Zeidler \cite{CLZ23} after introducing some technical modifications. 

We now state the first main result of this paper, a spacetime positive mass theorem for asymptotically flat initial data sets with arbitrary ends and a non-compact boundary under the tilted boundary dominant energy condition. The technical terms are defined in Section \ref{sec:pre}.
\begin{theorem}\label{thm:chairesultwithbends}
Let $(M^{n\geq 3},g,k,\mathcal{E},\Sigma)$ be an $n$-dimensional complete asymptotically flat spin initial data set with \textbf{arbitrary ends} and a non-compact boundary such that the interior dominant energy condition $\mu-|J|\geq 0$ and the tilted boundary dominant energy condition $H_{\Sigma} \pm \cos\alpha\, {\rm tr}_{\Sigma}k \geq \sin\alpha\, |k(\eta, \cdot)^{\top}|$ for some $\alpha\in [0,\frac{\pi}{2}]$. Then $E_{\mathcal{E}} \pm \cos\alpha\, (P_{\mathcal{E}})_n \geq \sin \alpha\, |\widehat{P}_{\mathcal{E}}|$.
\end{theorem}
As a direct consequence of Theorem \ref{thm:chairesultwithbends}, we have
\begin{corollary}
Let $(M^{n\geq 3},g,k,\mathcal{E},\Sigma)$ be an $n$-dimensional complete asymptotically flat spin initial data set with \textbf{arbitrary ends} and a non-compact boundary such that the interior dominant energy condition $\mu-|J|\geq 0$. Then the following holds:
\begin{itemize}
\item[(1)] 
if the tangential boundary dominant energy condition $H_{\Sigma} \geq |k(\eta,\cdot)^{\top}|$, then $E_{\mathcal{E}} \geq |\widehat{P}_{\mathcal{E}}|$,
\item[(2)]
if the normal boundary dominant energy condition $\theta_{\Sigma}^{\pm} \geq 0$, then $E_{\mathcal{E}} \pm (P_{\mathcal{E}})_{n} \geq 0$.
\end{itemize}
\end{corollary}
\begin{remark}
We would like to point the reader to the paper \cite{Liu+} for a proof of this result in dimension three using a harmonic level set approach.
\end{remark}
The non-spin case of Theorem \ref{thm:chairesultwithbends} is still unknown, but we also think that as in the spin case should hold. So we propose the following open conjecture.
\begin{conjecture}\label{conj:flat}
Theorem \ref{thm:chairesultwithbends} holds even if $M$ is not spin.
\end{conjecture}
In fact, with a little more care, we prove the strengthening of Theorem \ref{thm:chairesultwithbends}. 
\begin{theorem}\label{thm:flat--gener-PTM-arends}
Let $(\mathcal{E}^{n\geq 3},g,k)$ be an $n$-dimensional asymptotically flat initial data end with a non-compact boundary $\Sigma$ such that $E_{\mathcal{E}} \pm \cos\alpha\, (P_{\mathcal{E}})_n < \sin \alpha\, |\widehat{P}_{\mathcal{E}}|$ for some $\alpha\in [0,\frac{\pi}{2}]$. 
If $(M,g,k)$ is an $n$-dimensional initial data set contains $(\mathcal{E},g,k)$ as an open subset, then there exists an open neighborhood $\mathcal{U}$ around $\mathcal{E}$ in $M$ such that at least one of the following conditions must be violated:
\begin{itemize}
\item[(1)] $\overline{\mathcal{U}}$ is complete,
\item[(2)] $\mu-|J|\geq 0$ on $\mathcal{U}$,
\item[(3)] $H_{\Sigma}\pm \cos\alpha\, {\rm tr}_{\Sigma}k \geq \sin\alpha\, |k(\eta, \cdot)^{\top}|$,  
\item[(4)] $\mathcal{U}$ is spin.
\end{itemize}
\end{theorem}
We may use the same strategy to prove another main theorem of our paper for (unnecessarily complete) asymptotically flat initial data sets $(M,g,k)$ provided the tilted boundary dominant energy condition holds. 
In a sense, the occurrence of a shield of dominant energy (i.e., a subset of $M$ on which the dominant energy scalar $\mu-|J|$ has a positive lower bound) is a compromise of the incompleteness of $(M,g)$.
In the following result, the shield of dominant energy is a neck region on which the length of the neck is intimately related to the dominant energy scalar.
\begin{theorem}[Quantitative shielding theorem]\label{thm:PMT-with-dimension-zero-submanifolds}
Let $(M^{n \geq 3}, g, k)$ be an $n$-dimensional asymptotically flat initial data set with non-compact boundary, not assumed to be complete. Let $U_0, U_1$ be neighborhoods of an asymptotically flat end $\mathcal{E}$ with a non-compact boundary $\Sigma$ such that $\overline{U}_1\subset U_0$ and the closure $\overline{U_0\setminus \mathcal{E}}$ is compact. Suppose that $U_0$ is spin, and that $\overline{U}_0$ in $(M,g)$ is a complete manifold with boundary. Moreover, we assume that
\begin{itemize}
\item[(1)] the interior dominant energy condition $\mu-|J|\geq 0$ on $U_0$,
\item[(2)] the tilted boundary dominant energy condition
\begin{equation}
H_{\Sigma}\pm \cos\alpha\, {\rm tr}_{\Sigma}k \geq \sin\alpha\, |k(\eta, \cdot)^{\top}|
\end{equation}
on $\Sigma$ for some $\alpha\in [0,\frac{\pi}{2}]$,
\item[(3)] $\mu-|J|\geq \mathcal{Q}>0$ on $U_0\setminus \overline{U}_1$,
\item[(4)] for any fixed $d_0>0$, ${\rm dist}_g(\partial \overline{U}_1\setminus\Sigma, \partial \overline{U}_0\setminus\Sigma)> \sqrt{\frac{d_0}{\mathcal{Q}}}$.
\end{itemize}
Then $E_{\mathcal{E}} \pm \cos\alpha\, (P_{\mathcal{E}})_n \geq \sin \alpha\, |\widehat{P}_{\mathcal{E}}|$.
\end{theorem}
\begin{remark}
It should be noted that Theorem \ref{thm:PMT-with-dimension-zero-submanifolds} allows for the possibilities of incompleteness of $(M,g)$ and negativity of the dominant energy scalar, so long as this incompleteness and negativity of the dominant energy scalar are ``shielded'' from $U_0$. Moreover, we note that the hypotheses (3) and (4) of Theorem \ref{thm:PMT-with-dimension-zero-submanifolds} are equivalent to 
\begin{equation}\label{eq:shieldDN}
\mu-|J| > \frac{d_0}{ {\rm dist}_g^2(\partial \overline{U}_1\setminus\Sigma,\partial \overline{U}_0\setminus\Sigma)} \quad \text{on}\ U_0\setminus \overline{U}_1.
\end{equation}
It will be apparent from the proof of our quantitative shielding theorem that we can also accommodate the case when $(M,g,k)$ is a usual asymptotically flat spin initial data set. Therefore, in a way, the time-symmetric (i.e., $k\equiv 0$) case of inequality \eqref{eq:shieldDN} improves and sharpens the largeness assumptions about scalar curvature given in \cite[(1.3)]{LUY21} and \cite[(1.1)]{LLU22}. They proved analogs of our quantitative shielding theorem for asymptotically flat manifolds of dimensions $\leq 7$ via $\mu$-bubbles, except for the mass formula relating to the spectral estimates of Dirac-Witten operators with Callias potential, which has no dimensional constraint.
\begin{remark}
In contrast with the dominant energy shield \cite[Definition~1.13]{CLZ23} in the spin setting, the assumptions (3) and (4) in Theorem \ref{thm:PMT-with-dimension-zero-submanifolds} guarantee that we can construct a sequence of relatively easy and expectable Callias potential functions to relax the (outward) inward null expansion assumption of the compact boundary of the underlying initial data sets.
\end{remark}
\end{remark}
We emphasize that a key step in our approach is to find a solution to the mixed boundary value problem for the Dirac-Witten operator with a Callias potential on a complete connected asymptotically flat spin initial data set with a non-compact boundary and compact boundary (see Proposition \ref{pro:existence-lemma}). 
In order to solve such boundary value problems, a generalization with the mixed boundary conditions of the particular case of a theorem of Grosse and Nakad \cite{GN14} is enough to guarantee the existence of their solutions, cf. Theorem \ref{thm:appendix-thm}.
In the use of these solutions, we can obtain a sequence of spacetime harmonic spinors (i.e., the elements in the kernel of Dirac-Witten operators with a sequence of Callias potentials) for spinorial proofs of our main theorems, which are all asymptotic to a constant section of the relative Dirac bundle (see Subsection \ref{subsec:Dirac-Witten}).

This paper is organized as follows. In Section \ref{sec:pre}, we introduce some backgrounds and facts of technical preparations. In Section \ref{sec:flat}, we prove Theorem \ref{thm:flat--gener-PTM-arends} and Theorem \ref{thm:PMT-with-dimension-zero-submanifolds}. An essential mixed boundary value problem associated with the Dirac-Witten operator with a Callias potential is included in appendix \ref{sec:appendix}.

\section*{Acknowledgement}
\vspace{-.1in}
The author would like to sincerely thank Bo Liu (ECNU, Shanghai)  and Zhenlei Zhang (CNU, Beijing) for their constant encouragement and valuable discussions concerning an earlier version of this manuscript. The author is also very grateful to Xueyuan Wan (CQUT, Chongqing) for his helpful discussions and interest in this work. Additional thanks go to Pengshuai Shi (BIT, Beijing) for communicating the paper \cite{BB12} and Xiaoxiang Chai (POSTECH, Pohang) for the explanation of the case $\alpha=0$ of Theorem~\ref{thm:chairesultwithbends} in his preprint \cite{Cha23}.

\section{Preliminaries}\label{sec:pre}
In this section, we utilize the techniques of the Dirac-Witten operator with a Callias potential to study the ADM energy-momentum of asymptotically flat initial data sets with arbitrary ends and a non-compact boundary.
Such Callias-type operators were first used in \cite{Cec20,CZ21a,Zei20,Zei22} within the context of the time-symmetric version of three problems proposed by Gromov \cite{Gro18,Gro19}: the long neck problem, the band width inequality, and the width estimates of geodesic collar neighborhood of boundary. Applications to the spectral and non-time-symmetric cases of these problems were recently handled by Hirsch, Kazaras, Khuri, and Zhang \cite{HKKZ23} and by the author \cite{Liu23a,Liu23b}.
\subsection{Asymptotically flat initial data sets with arbitrary ends and a non-compact boundary}
Let us start by adapting well-known definitions, see e.g., \cite{Lee19}.
\begin{definition}
An \textit{initial data set} $(M, g, k)$ is a Riemannian manifold $(M,g)$ with a symmetric $(0,2)$-tensor $k$. The \textit{energy density} and \textit{momentum density} are respectively given by 
\begin{equation}
   \mu = \frac{1}{2} ({\rm scal}_g  + ({\rm tr}_g k)^2 - |k|_g^2),\quad J= {\rm div}_g k -{\rm d}({\rm tr}_g k),
\end{equation} 
where ${\rm scal}_g$ is the scalar curvature of $g$.
\end{definition}
\begin{definition}
We say that $(M,g,k)$ satisfies the \textit{interior dominant energy condition} if
\begin{equation}
\mu - |J| \geq 0
\end{equation}
everywhere along $M$.
\end{definition}
If $\partial M\neq \emptyset$, let $\eta$ be the outward unit normal of $\partial M$ in $M$, then the \textit{mean curvature} of $\partial M$ with respect to $\eta$ is defined as
\begin{equation}
    H_{\partial M} = {\rm tr}_{\partial M}(\nabla \eta)= \sum_{i=1}^{n-1}\langle e_i,\nabla_{e_i}\eta \rangle .
\end{equation}
where $\nabla \eta$ is the shape operator and $e_1,\cdots, e_{n-1}$ is a local orthonormal frame on $\partial M$.
\begin{definition}
Given an initial data set $(M,g,k)$ potentially with boundary $\partial M$, the \textit{outward (inward) null expansion} of a two-sided hypersurface $F$ in $M$ is defined by
\begin{equation}
\theta_{F}^{\pm} = H_{F} \pm {\rm tr}_{F}k.
\end{equation}
\end{definition}
\begin{definition}
We say that $(M,g,k)$ satisfies the \textit{normal boundary dominant energy condition} if
\begin{equation}\label{eq:nBDEC}
\theta_{\partial M}^{\pm} \geq 0
\end{equation}
everywhere along $\partial M$.
\end{definition}
In {\cite{AdLM21}}, Almaraz, de Lima, and Mari coped with initial data sets with boundary by introducing a so-called tangential boundary dominant energy condition.
\begin{definition}
We say that $(M,g,k)$ satisfies the \textit{tangential boundary dominant energy condition} if
\begin{equation}
H_{\partial M} \geq |k(\eta,\cdot)^{\top}|
\end{equation}
everywhere along $\partial M$, where $k(\eta,\cdot)^{\top}$ is the component of the $1$-form $k(\eta,\cdot)$ tangential to $\partial M$. 
\end{definition}
The following tilted boundary dominant energy condition in Chai \cite{Cha23} generalizes the tangential and normal boundary dominant energy conditions.
\begin{definition}
We say that $(M,g,k)$ satisfies the \textit{tilted boundary dominant energy condition} if
\begin{equation}\label{eq:defn-tilted}
H_{\partial M} \pm \cos\alpha\, {\rm tr}_{\partial M} k \geq \sin\alpha\, |k(\eta,\cdot)^{\top}|
\end{equation}
everywhere along $\partial M$, where $\alpha\in [0,\frac{\pi}{2}]$ is a constant angle and $k(\eta,\cdot)^{\top}$ is the component of the $1$-form $k(\eta,\cdot)$ tangential to $\partial M$.
\end{definition}
With the spirit to \cite{AdLM21,LUY21}, we may consider the following concept of an asymptotically flat initial data set with arbitrary ends and a non-compact boundary.
\begin{definition}
An $n$-dimensional initial data set $(M,g,k,\mathcal{E},\Sigma)$ with a non-compact boundary $\Sigma$ contained in a distinguished asymptotically flat initial data end $\mathcal{E}\subseteq M$ is called an \textit{asymptotically flat initial data set with arbitrary ends and a non-compact boundary} if 
$\mu$ and $J$ are integrable in $\mathcal{E}$, $H_{\Sigma}$ and $|k(\eta,\cdot)^{\top}|$ are integrable on $\Sigma$, and there exist $r_0>0$ and a diffeomorphism
\begin{equation}
\Phi:\mathcal{E} \overset{\cong}{\to} \mathbb{R}^n_{+,r_0}=\{x\in \mathbb{R}^n_{+}; r > r_0\}
\end{equation}
such that as $r \to \infty$, we have
\begin{equation}
|g_{ij}-\delta_{ij}|_{\delta}+ r |\partial g_{ij}|_{\delta} + r^2 |\partial^2 g_{ij}|_{\delta} = O(r^{-\tau}), \quad  |k_{ij}|_{\delta} + r |\partial k_{ij}|_{\delta} = O(r^{-\tau-1})
\end{equation}
for all $1\leq i,j\leq n$ and $\tau>\frac{n-2}{2}$, where $x=(x_1,\cdots,x_n)$ is the coordinate chart induced by $\Phi$, $r=|x|$, $\mathbb{R}^n_{+}$ is the closed Euclidean half-space with the standard flat metric $\delta$, and we have identified $g$ and $k$ with their pull-backs under $\Phi^{-1}$ for simplicity of notation. Here the phrase ``arbitrary ends'' means that one does not need to impose additional requirements on those ends other than $\mathcal{E}$.  
\end{definition}
Under these asymptotic expansions, we may assign to $(M,g,k,\mathcal{E},\Sigma)$ an ADM energy-momentum type asymptotic invariant as follows.
\begin{definition}[{\cite{AdLM21}}]
The \textit{ADM energy-momentum vector} $(E_{\mathcal{E}},P_{\mathcal{E}})$ of a distinguished asymptotically flat $\mathcal{E}$ with a non-compact boundary $\Sigma$ is defined as
\begin{equation}\label{defn:energy}
E_{\mathcal{E}} =\lim_{r\to \infty} \left\{ \int_{S^{n-1}_{r,+}} (g_{ij,j}-g_{jj,i})\nu^i dS^{n-1}_{r,+} - \int_{S^{n-2}_{r}} g_{\alpha n}\vartheta^{\alpha} dS^{n-2}_{r} \right\},
\end{equation}
and
\begin{equation}\label{defn:momentum}
(P_{\mathcal{E}})_{i} = \lim_{r\to \infty} 2 \int_{S^{n-1}_{r,+}} \pi_{ij} \nu^j dS^{n-1}_{r,+},
\end{equation}
where $S^{n-1}_{r,+} \subset \mathcal{E}$ is a large coordinate hemisphere of radius $r$ with outward unit normal $\nu$, and $\vartheta=\left.\nu\right|_{S^{n-2}_{r}}$ is the outward unit co-normal to $S^{n-2}_{r}=\partial S^{n-1}_{r,+}$, viewed as the boundary of the bounded region $\Sigma_{r}\subset \Sigma$ (see Figure \ref{figure:AF-arbitrary-ends-noncompact-boundary}), $\pi=k-({\rm tr}_{g}k)g$, $dS^{n-1}_{r,+}$ and $dS^{n-2}_{r}$ denote respectively the volume element on $S^{n-1}_{r,+}$ and $S^{n-2}_{r}$ with respect to the background Euclidean metric.
\end{definition}
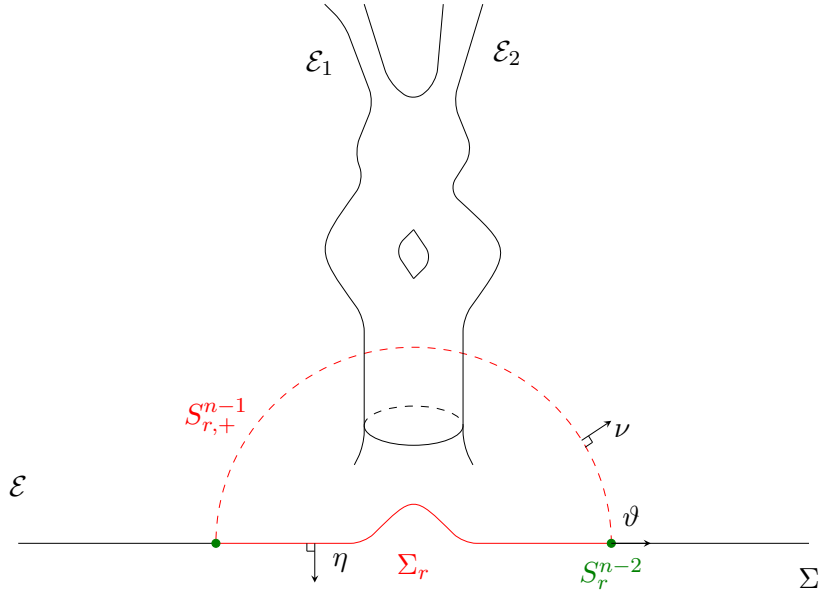
\begin{figure}[H]
   \centering
\begin{tikzpicture}[scale=1.3]
\draw[dashed,red,line width=0.2pt] (2,0) arc [start angle=0, end angle=180, x radius=2, y radius=2] node[above=1.3cm]{$S^{n-1}_{r,+}$};
\draw[rounded corners=5,line width=0.2pt] (-4,0) node[above=0.5cm]{$\mathcal{E}$}--(-2,0);
\draw[rounded corners=5,line width=0.2pt] (4,0) node[below=.2cm]{$\Sigma$} --(2,0);
\draw[rounded corners=5,red,fill=white,line width=0.2pt] (-2,0)--(-0.5,0).. controls (0,0.5) ..node[below=.5cm]{$\Sigma_r$} (0.5,0)--(2,0);
\draw[rounded corners=5,line width=0.2pt] (-0.6,0.8)--(-0.5,1)--(-0.5,2.3).. controls (-1,3) ..(-0.5,3.7)--(-0.6,4)--(-0.4,4.5)--node[left=.2cm]{$\mathcal{E}_1$} (-0.7,5.3)--(-0.9,5.5);
\draw[rounded corners=5,line width=0.2pt] (0.6,0.8)--(0.5,1)--(0.5,2.3).. controls (1,3) ..(0.35,3.6)--(0.6,4)--(0.4,4.5)--node[right=.2cm]{$\mathcal{E}_2$} (0.7,5.5);
\draw[rounded corners=5,line width=0.2pt] (-0.5,5.5)--(-0.25,4.7)--(0,4.5)--(0.23,4.7)--(0.3,5.5);
\draw[rounded corners=5,line width=0.2pt] (0,3.2)--(-0.2,3)--(0,2.7);
\draw[rounded corners=5,line width=0.2pt] (0,3.2)--(0.2,2.9)--(0,2.7);
\draw[dashed,line width=0.2pt] (-0.5,1.2) arc [start angle=180, end angle=0, x radius=0.5, y radius=0.2];
\draw[line width=0.2pt] (0.5,1.2) arc [start angle=0, end angle=-180, x radius=0.5, y radius=0.2];
\node[label=270:\textcolor{green!50!black}{$S^{n-2}_r$},green!50!black,circle, draw, fill,minimum size=3pt,inner sep=0pt] at (2,0) {};
\node[green!50!black,circle, draw, fill,minimum size=3pt,inner sep=0pt] at (-2,0) {};
\draw[-stealth,line width=0.2pt] (2,0)--node[above=0.1cm]{$\vartheta$}(2.4,0);
\draw[-stealth,line width=0.2pt] (1.7,1.05)--node[right=0.1cm]{$\nu$}(2,1.25);
\draw[-,line width=0.2pt] (1.735,0.98)--(1.81,1.03);
\draw[-,line width=0.2pt] (1.81,1.03)--(1.77,1.095);
\draw[-stealth,line width=0.2pt] (-1,0)--node[right=0.1cm]{$\eta$}(-1,-0.4);
\draw[-,line width=0.2pt] (-1.08,0)--(-1.08,-0.08);
\draw[-,line width=0.2pt] (-1.08,-0.08)--(-1,-0.08);
\end{tikzpicture}
\caption{An asymptotically flat initial data set with arbitrary ends $\mathcal{E}_1,\mathcal{E}_2$ and a non-compact boundary $\Sigma$ contained in a distinguished asymptotically flat initial data end $\mathcal{E}$.} \label{figure:AF-arbitrary-ends-noncompact-boundary}
\end{figure}
As in \cite{Cha23}, we denote $\widehat{P}_{\mathcal{E}}=((P_{\mathcal{E}})_1,\cdots, (P_{\mathcal{E}})_{n-1})$ and we conclude $(P_{\mathcal{E}})_{n}$ in the ADM momentum. This is different from \cite{AdLM21}. Note that $(E_{\mathcal{E}},P_{\mathcal{E}})$ is well defined, that is, the limits on the right-hand side of \eqref{defn:energy} and \eqref{defn:momentum} exist and their values do not depend on the choice of the induced coordinate chart on the distinguished asymptotically flat initial data end $\mathcal{E}$, cf. \cite[Corollary~3.4]{AdLM21}.

\subsection{Dirac-Witten operators with Callias potential}\label{subsec:Dirac-Witten}
In this subsection, we recall some basic facts about the Dirac-Witten operator with a Callias potential, cf. \cite{CZ21a}.

Let $(X,g,k)$ be a spin initial data set possibly with boundary. Let $\mathcal{S}_{X} \to X$ be the complex spinor bundle on $X$. Then $S:=\mathcal{S}_{X}\oplus\mathcal{S}_{X}$ becomes a $\mathbb{Z}_2$-graded Dirac bundle with the induced connection $\nabla=\nabla_{\mathcal{S}_{X}}\oplus \nabla_{\mathcal{S}_{X}}$ and the Clifford multiplication
\begin{equation}
     c(\xi)=\begin{pmatrix}
       0 & c_{\mathcal{S}}(\xi) \\
       c_{\mathcal{S}}(\xi) & 0
     \end{pmatrix},
\end{equation}
where $c_{\mathcal{S}}$ and $\nabla_{\mathcal{S}_{X}}$ are respectively the Clifford multiplication and connection on $\mathcal{S}_{X}$. Actually, $S$ is a relative Dirac bundle in the sense of Cecchini-Zeilder \cite{CZ21a} with the involution
\begin{equation}
     \sigma=\begin{pmatrix}
     0&-\sqrt{-1}\\
     \sqrt{-1} &0
     \end{pmatrix}.
\end{equation}
Let $\slashed{D}:C^{\infty}(X,\mathcal{S}_{X})\to C^{\infty}(X,\mathcal{S}_{X})$ is the spinor Dirac operator on $(X,g)$, where $C^\infty(X,\mathcal{S}_{X})$ denotes the space of all smooth sections of $\mathcal{S}_{X}$ over $X$. Then the Dirac operator on $S$ is given by
\begin{equation}
\mathcal{D}=\begin{pmatrix}
0&\slashed{D}\\
\slashed{D}&0
\end{pmatrix}.
\end{equation}
We define a new connection on $S$ by 
\begin{equation}
     \widetilde{\nabla}_{e_i} = \nabla_{e_i} + \frac{1}{2} k_{ij} c(e^j)\sigma,
\end{equation}
then the associated \textit{Dirac-Witten operator} is defined as
\begin{equation}
     \widetilde{\mathcal{D}} = c(e^i)\widetilde{\nabla}_{e_i}=\mathcal{D}-\frac{{\rm tr}_{g}k}{2}\sigma.
\end{equation}
For a Lipschitz function $f$ with compact support on $X$, we define the modified connection on $S$ by
\begin{equation}
    \widetilde{\nabla}_{e_i}^{f} = \nabla_{e_i} + \frac{1}{2} k_{ij} c(e^j)\sigma - \frac{f}{n}c(e^i)\sigma,
\end{equation}
then the associated \textit{Dirac-Witten operator with a Callias potential} (cf. \cite{CLZ23}) is defined as
\begin{equation}
     \widetilde{\mathcal{D}}^{f} = c(e^i)\widetilde{\nabla}_{e_i}^{f}=\mathcal{D}-\frac{{\rm tr}_{g}k}{2}\sigma + f\sigma =\widetilde{\mathcal{D}} + f\sigma.
\end{equation}
By direct calculations, we find the following basic fact (see e.g., \cite[Proposition~3.5]{CLZ23}).
\begin{lemma}
Let $X$ be a spin initial data set with boundary and let $\eta$ be the outward unit normal of $\partial X$. Let $C_c^{\infty}(X,S)$ denote the space of all smooth sections with compact support of $S$ over $X$. Then for any $u\in C_c^{\infty}(X,S)$, we have 
\begin{equation}\label{eq:CLZ23-spectral-estimate}
\begin{aligned}
\int_{X} |\widetilde{\mathcal{D}}^{f}u|^2 dV = & \int_{X} |\widetilde{\nabla}^{f}u|^2 dV  + \frac{1}{2} \int_{X} \langle (\mu + J_{i}c(e^i)\sigma ) u, u\rangle dV  \\
& + \frac{n-1}{n} \int_{X} \langle (f^2 +c({\rm d}f)\sigma) u, u\rangle dV \\
& - \int_{\partial X}\langle c(\eta^{\flat})\widetilde{\mathcal{D}} u + \widetilde{\nabla}_{\eta} u+ \frac{n-1}{n}f c(\eta^{\flat})\sigma u, u\rangle dA. 
\end{aligned}
\end{equation}
\end{lemma}

\subsection{ADM energy-momentum formula}
In this subsection, we use Proposition \ref{pro:existence-lemma} to show the relation between the ADM energy-momentum and spectral estimates of the Dirac-Witten operator with a Callias potential. Such ADM energy-momentum formula constitutes an essential ingredient of the proofs of our main theorems.

Let $(X,g,k)$ be a complete connected asymptotically flat spin initial data set with a non-compact boundary $\Sigma$ and compact boundary $\partial X\setminus \Sigma$. Assume that $M$ only has an asymptotically flat initial data end $\mathcal{E}$ with a non-compact boundary $\Sigma$. Let $\eta$ be the outward unit normal of $\partial X$. 
Let $f:X\to \mathbb{R}$ be a Lipschitz Callias potential function with compact support, which is disjointed with $\Sigma$. 
In order to obtain a spectral estimate inequality, our general strategy is to approximate $X$ by a family of compact connected regions $M_r$ with boundary, that is, $M_r$ becomes $X$ when $r\to \infty$. Then we apply \eqref{eq:CLZ23-spectral-estimate} to each of these approximations $M_r$, and finally we get the desired inequality by taking the limit.
Note that the boundary $\partial M_r=S^{n-1}_{r,+}\cup \Sigma_{r}\cup (\partial M_{r} \setminus (S^{n-1}_{r,+} \cup \Sigma_{r}))$, where $\Sigma_r=\partial M_r\cap \Sigma$, $S^{n-1}_{r,+}$ denotes the large coordinate hemisphere in the asymptotically flat initial data end $\mathcal{E}$, and $\partial M_{r} \setminus (S^{n-1}_{r,+} \cup \Sigma_{r})$ denotes the part of $\partial M_r$ independent of $r$ (see Figure \ref{fig:regionMr}). 
For any $u\in C^\infty(X,S)$, from \eqref{eq:CLZ23-spectral-estimate} we have
\begin{figure}
\begin{center}
\begin{tikzpicture}[scale=1]
\coordinate (A) at (0,3);
\coordinate (B) at (0,-3);
\coordinate (C) at (0,1.3);
\coordinate (D) at (0,-1.3);
\coordinate (E) at (0.5,0.5);
\coordinate (F) at (4.3,2.5);
\coordinate (G) at (0.5,-0.5);
\coordinate (H) at (2.9,-2.5);
\coordinate (I) at (5.05,2.5);
\coordinate (J) at (6.8,0.13);
\coordinate (K) at (3.7,-2.5);
\coordinate (L) at (6.75,-0.63);
\draw[dashed,red,fill=blue!5,line width=0.3pt] (C) arc [start angle=90, end angle=-90, x radius=1.3, y radius=1.3];
\draw[line width=0.2pt,fill=blue!5] (E) to [bend right=40] (F) to [bend left=23] (I) to [bend right=40] (J) to [bend left=23] (L) to [bend right=23] (K) to [bend left=23] (H) to [bend right=23] (G);
\draw[dashed,red,line width=0.3pt] (C) arc [start angle=90, end angle=-90, x radius=1.3, y radius=1.3];
\draw[rounded corners=5,line width=0.2pt] (A) to (C);
\draw[rounded corners=5,line width=0.2pt] (B) to (D);
\draw[rounded corners=5,red,fill=white,line width=0.3pt] (C) to (0,0.6) to (0.25,0.4) to (0,0.2) to (0,-0.5) to (0.45,-0.55) to (0,-0.7) to (D);
\draw[-stealth,line width=0.2pt] (0,1) to (-0.4,1);
\draw[line width=0.2pt] (-0.1,1) to (-0.1,0.9) to (0,0.9);
\draw[-stealth,line width=0.2pt] (D) to (0,-1.7);
\draw[-stealth,line width=0.2pt] (0.7,-1.1) to (0.86,-1.43);
\draw[line width=0.2pt] (0.628,-1.14) to (0.66,-1.21) to (0.74,-1.17);
\draw[rounded corners=5,blue,line width=0.3pt] (L) to [bend right=23] (J);
\draw[rounded corners=5,blue,line width=0.3pt] (L) to [bend left=23] (J);
\draw[rounded corners=5,blue,line width=0.3pt] (H) to [bend right=23] (K);
\draw[rounded corners=5,blue,line width=0.3pt] (H) to [bend left=23] (K);
\draw[rounded corners=5,blue,line width=0.3pt] (F) to [bend right=23] (I);
\draw[rounded corners=5,blue,line width=0.3pt] (F) to [bend left=23] (I);
\coordinate (M) at (4,0.5);
\coordinate (N) at (4.3,-0.5);
\coordinate (O) at (4.13,0.47);
\coordinate (P) at (4.35,-0.45);
\draw[fill=white,rounded corners=5,line width=0.2pt] (M) to [bend left=70] (N);
\draw[fill=white,rounded corners=5,line width=0.2pt] (O) to [bend right=60] (P);
\node at (1.2,1.4) {\(\textcolor{red}{S^{n-1}_{r,+}}\)};
\node at (-0.7,-3) {\(\Sigma\)};
\node at (1,2.8) {\(\mathcal{E}\)};
\node at (-0.7,-1.3) {\textcolor{green!50!black}{\(S^{n-2}_{r}\)}};
\node at (8.7,-0.3)  {\(\textcolor{blue}{\partial M_{r}\setminus (S^{n-1}_{r,+}\cup \Sigma_r)}\)};
\node at (3,-0.2)  {\(M_r\)};
\node at (5.5,-1.7) {\(X\)};
\node at (-0.6,0)  {\(\textcolor{red}{\Sigma_r}\)};
\node at (-0.7,1)  {\(\eta\)};
\node at (0.2,-1.7)  {\(\vartheta\)};
\node at (1.05,-1.5) {\(\nu\)};
\node[circle,fill=green!50!black,minimum size=3pt,inner sep=0pt] at (C) {};
\node[circle,fill=green!50!black,minimum size=3pt,inner sep=0pt] at (D) {};
\end{tikzpicture}
\caption{The region $M_r$ in $X$.}\label{fig:regionMr}
\end{center}
\end{figure}
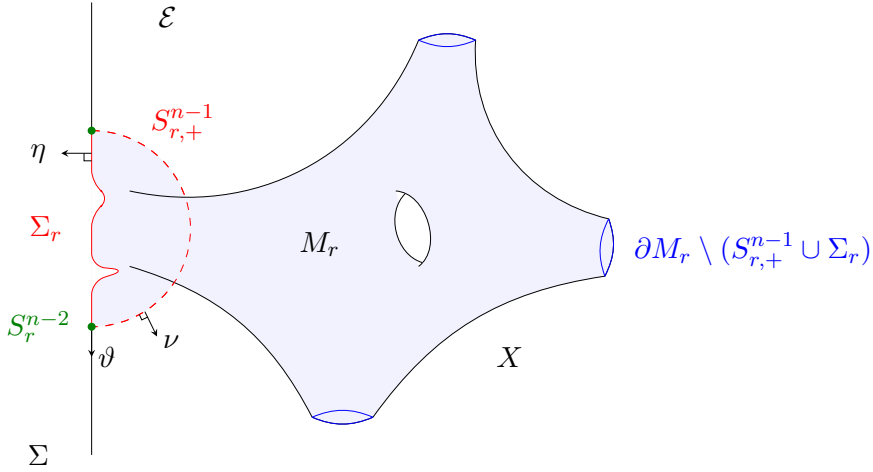
\begin{equation}\label{eq:mass-spectral-estimate}
\begin{aligned}
\int_{M_r} |\widetilde{\mathcal{D}}^{f} & u|^2 dV =\int_{M_r} |\widetilde{\nabla}^{f}u|^{2} dV  + \frac{1}{2} \int_{M_r} \langle (\mu + J_{i}c(e^i)\sigma ) u, u\rangle dV \\
&+ \frac{n-1}{n} \int_{M_r} \langle (f^2 +c({\rm d}f)\sigma) u, u\rangle dV \\
&- \int_{\partial M_r\setminus S^{n-1}_{r,+} }\langle c(\eta^{\flat})\widetilde{\mathcal{D}} u + \widetilde{\nabla}_{\eta} u+ \frac{n-1}{n}f c(\eta^{\flat})\sigma u, u\rangle dA \\
&-\int_{S^{n-1}_{r,+}}\langle c(\nu^{\flat})\widetilde{\mathcal{D}} u + \widetilde{\nabla}_{\nu} u, u\rangle dA - \frac{n-1}{n} \underbrace{ \int_{S^{n-1}_{r,+} }\langle f c(\nu^{\flat})\sigma u, u\rangle dA }_{= 0 \, \text{for}\, r\gg 1},
\end{aligned}
\end{equation}
where $\nu$ is the unit normal of $S^{n-1}_{r,+}$ pointing to the infinity.

Consider the Dirac bundle $\left.S\right|_{\partial M_r\setminus S^{n-1}_{r,+}}$ with the Clifford multiplication and connection
\begin{equation}
     c^{\partial}(e^i)=-c(e^i)c(\eta^{\flat}),\quad \nabla^{\partial}_{e_i}=\nabla_{e_i}-\frac{1}{2}c^{\partial}(\nabla_{e_i}\eta^{\flat}).
\end{equation}
The corresponding boundary Dirac operator 
\begin{equation}
     \mathcal{A}: C^{\infty}(\partial M_r\setminus S^{n-1}_{r,+}, \left.S\right|_{\partial M_r\setminus S^{n-1}_{r,+}}) \to  C^{\infty}(\partial M_r\setminus S^{n-1}_{r,+}, \left.S\right|_{\partial M_r\setminus S^{n-1}_{r,+}}) 
\end{equation}
is given by
\begin{equation}
\mathcal{A}:= \sum_{i=1}^{n-1} c^{\partial}(e^i) \nabla^{\partial}_{e_i}.
\end{equation}
Then we have
\begin{equation}\label{eq:boundary-Dirac-identity}
\begin{aligned}
\mathcal{A} & = \frac{1}{2} H_{\partial M_r\setminus S^{n-1}_{r,+} } + c(\eta^{\flat}) \mathcal{D} + \nabla_{\eta}\\
&= \frac{1}{2} H_{\partial M_r\setminus S^{n-1}_{r,+}} + c(\eta^{\flat})(\widetilde{\mathcal{D}} + \frac{{\rm tr}_{g} k}{2}\sigma ) + \widetilde{\nabla}_{\eta} - \frac{1}{2} k_{j\eta} c(e^j) \sigma \\
&= \frac{1}{2} H_{\partial M_r\setminus S^{n-1}_{r,+}}  + \frac{{\rm tr}_{g} k}{2}c(\eta^{\flat}) \sigma - \frac{1}{2} k_{\eta\eta}c(\eta^{\flat})\sigma - \frac{1}{2} k_{a\eta} c(e^a)\sigma \\
&\quad  + c(\eta^{\flat})\widetilde{\mathcal{D}}+\widetilde{\nabla}_{\eta}\\
&= \frac{1}{2} H_{\partial M_r\setminus S^{n-1}_{r,+}}  + \frac{{\rm tr}_{\partial M_r\setminus S^{n-1}_{r,+}} k}{2}c(\eta^{\flat}) \sigma - \frac{1}{2} k_{a\eta} c(e^a)\sigma + c(\eta^{\flat})\widetilde{\mathcal{D}}+\widetilde{\nabla}_{\eta}.
\end{aligned}
\end{equation}
\begin{definition}
We say that the endomorphism $\chi: \left. S\right|_{\partial M_{r} \setminus ( S^{n-1}_{r,+} \cup \Sigma_{r})} \to \left. S\right|_{\partial M_{r} \setminus ( S^{n-1}_{r,+} \cup \Sigma_{r})}$ given by
\begin{equation}
\chi u=-c(\eta^{\flat}) \sigma u
\end{equation}
is the \textit{chirality operator} on $\left.S\right|_{\partial M_{r} \setminus ( S^{n-1}_{r,+} \cup \Sigma_{r})}$. 
\end{definition}
By assuming the boundary condition
\begin{equation}\label{eq:BC-chi}
\chi \left.u\right|_{\partial M_{r} \setminus ( S^{n-1}_{r,+} \cup \Sigma_{r})}= \left.u\right|_{\partial M_{r} \setminus ( S^{n-1}_{r,+} \cup \Sigma_{r})},
\end{equation}
we obtain
\begin{equation}\label{eq:chi-1}
\langle \left. u\right|_{\partial M_{r} \setminus ( S^{n-1}_{r,+} \cup \Sigma_{r})}, \mathcal{A} \left. u\right|_{\partial M_{r} \setminus ( S^{n-1}_{r,+} \cup \Sigma_{r})} \rangle = 0,
\end{equation}
\begin{equation}\label{eq:chi-2}
\langle \left.u\right|_{\partial M_{r} \setminus ( S^{n-1}_{r,+} \cup \Sigma_{r})}, c(e^a) \sigma \left. u\right|_{\partial M_{r} \setminus ( S^{n-1}_{r,+} \cup \Sigma_{r})} \rangle = 0.
\end{equation}
Thus by \eqref{eq:nBDEC}, \eqref{eq:boundary-Dirac-identity}, \eqref{eq:chi-1}, and \eqref{eq:chi-2} under the assumption \eqref{eq:BC-chi} we have
\begin{equation}\label{eq:term-of-compact-boundary}
\begin{aligned}
& - \int_{\partial M_{r} \setminus ( S^{n-1}_{r,+} \cup \Sigma_{r})} \langle c(\eta^{\flat})\widetilde{\mathcal{D}}u + \widetilde{\nabla}_{\eta} u, u\rangle dA \\
= & \int_{\partial M_{r} \setminus ( S^{n-1}_{r,+} \cup \Sigma_{r})} \frac{1}{2} \theta_{\partial M_{r} \setminus ( S^{n-1}_{r,+} \cup \Sigma_{r})}^{-} |u|^2 dA.
\end{aligned}
\end{equation}
We now shall deal with the boundary condition for non-compact boundary $\Sigma$.
\begin{definition}[{\cite{AdL22,Cha23}}]
Let $\alpha\in [-\frac{\pi}{2},\frac{\pi}{2}]$. We say that the endomorphism $Q:\left.S\right|_{\Sigma} \to \left.S\right|_{\Sigma}$ given by
\begin{equation}\label{eq:defnQ}
Qu=\cos\alpha\, \sigma c(\eta^{\flat}) u + \sqrt{-1}\sin\alpha\, c(\eta^{\flat}) u
\end{equation}
is the \textit{chirality operator} on $\left.S\right|_{\Sigma}$.
\end{definition}
Note that the angle $\alpha\in [-\frac{\pi}{2},\frac{\pi}{2}]$ in \eqref{eq:defnQ} but $\alpha\in [0,\frac{\pi}{2}]$ in \eqref{eq:defn-tilted} since we need a choice on the sign of $\alpha$ later.
In particular, let $\overline{\mathcal{S}}$ be the standard spinor bundle over $(\mathbb{R}^n_{+},\delta,0)$. For the $\mathbb{Z}_2$-graded Dirac bundle $S=\overline{\mathcal{S}}\oplus \overline{\mathcal{S}}$, we define the chirality operator acting on $C^{\infty}(\partial \mathbb{R}^n_{+}, \left.S\right|_{\partial \mathbb{R}^n_{+}})$ by
\begin{equation}\label{eq:defn-Q-bar}
\overline{Q}u_0 =\cos\alpha\, \sigma c({\rm d}x^n)u_0 + \sqrt{-1} \sin \alpha\, c({\rm d}x^n) u_0
\end{equation}
for any smooth section $u_0$, where ${\rm d}x^n$ is the dual $1$-form of the outward unit normal of $\partial \mathbb{R}^n_{+}$. Then $\left.S\right|_{\partial \mathbb{R}^n_{+}}$ decomposes into $\pm 1$-eigenbundles in which the section $u_0$ satisfies respectively
\begin{equation}\label{eq:BC-Q-bar}
\overline{Q}u_0=\pm u_0.
\end{equation}
To proceed, we shall calculate the boundary term $\langle c(\eta^{\flat})\widetilde{\mathcal{D}}+ \widetilde{\nabla}_{\eta} u, u\rangle$ on $\Sigma$ under the boundary condition
\begin{equation}\label{eq:BC-Sigma}
Q \left.u\right|_{\Sigma}=\pm \left.u\right|_{\Sigma}.
\end{equation}
For a smooth section $u$ satisfying \eqref{eq:BC-Sigma} on $\Sigma$, by \cite[Lemma~3.5]{Cha23} we have
\begin{equation}\label{eq:cal-2}
\langle c(\eta^{\flat})\sigma \left.u\right|_{\Sigma}, \left.u\right|_{\Sigma}\rangle = \mp \cos \alpha\, \big|\left.u\right|_{\Sigma} \big|^2,
\end{equation}
\begin{equation}\label{eq:cal-3}
\langle c(e^{a})\sigma \left.u\right|_{\Sigma}, \left.u\right|_{\Sigma}\rangle = \mp \langle \sqrt{-1}\sin \alpha\, c(e^{a}) c(\eta^{\flat}) \sigma \left.u\right|_{\Sigma}, \left.u\right|_{\Sigma} \rangle.
\end{equation}
From \cite[(4.27)]{CH03}, we see that $\langle \mathcal{A} \left.u\right|_{\Sigma}, \left.u\right|_{\Sigma}\rangle =0$ on $\Sigma$. By combining \eqref{eq:cal-2}, \eqref{eq:cal-3}, and \eqref{eq:boundary-Dirac-identity}, we have (compare with \cite[Lemma~3.6]{Cha23}) 
\begin{equation}\label{eq:term-of-noncompact-boundary}
\begin{aligned}
&- \int_{\Sigma_r} \langle c(\eta^{\flat})\widetilde{\mathcal{D}}u + \widetilde{\nabla}_{\eta} u, u\rangle dA \\
=& \int_{\Sigma_r} \langle \big(\frac{1}{2}H_{\Sigma_r} + \frac{{\rm tr}_{\Sigma_r}k}{2}c(\eta^{\flat})\sigma-\frac{1}{2}k_{a\eta}c(e^{a})\sigma-\mathcal{A} \big) u,u  \rangle dA\\
=& \frac{1}{2} \int_{\Sigma_r} (H_{\Sigma_r}  \mp \cos\alpha\, {\rm tr}_{\Sigma_r}k) |u|^2 \pm  \sin \alpha\, \langle \sqrt{-1} k_{a\eta} c(e^{a}) c(\eta^{\flat}) \sigma u, u \rangle dA \\
\geq & \frac{1}{2} \int_{\Sigma_r} \Big( H_{\Sigma_r}\mp \cos\alpha\, {\rm tr}_{\Sigma_r}k -\sin\alpha\, (\sum_{a=1}^{n-1} k_{a\eta}^2)^{\frac{1}{2}} \Big) |u|^2 dA\\ 
\geq & \frac{1}{2} \int_{\Sigma_r} \left( H_{\Sigma_r}\mp \cos\alpha\, {\rm tr}_{\Sigma_r}k -\sin\alpha\, |k(\eta,\cdot)^{\top}| \right)|u|^2 dA.
\end{aligned}
\end{equation}
If $u$ is asymptotic to a constant section $u_0$ (compare \cite[Definition~2.4]{CZ21b} for precise definition) subject to \eqref{eq:BC-Q-bar}, the argument of \cite[Section~4]{Cha23} shows that
\begin{equation}\label{eq:mass1}
\begin{aligned}
&\int_{S^{n-1}_{r,+}}\langle c(\nu^{\flat})\widetilde{\mathcal{D}} u + \widetilde{\nabla}_{\nu} u, u\rangle dA \\
\to & \frac{1}{4}(E_{\mathcal{E}} \pm \cos\alpha\, (P_{\mathcal{E}})_n ) |u|^2 \pm \sin\alpha\, (P_{\mathcal{E}})_a \langle \sqrt{-1} c({\rm d}x^a) c({\rm d}x^n) \sigma u_0, u_0 \rangle_{\delta} 
\end{aligned}
\end{equation}
when $r\to \infty$. Furthermore, there exists a choice of $u_0$ and the sign of $\alpha$ such that the following identity holds (cf. \cite[Lemma~4.2]{Cha23}):
\begin{equation}\label{eq:mass2}
\mp \sin\alpha\, (P_{\mathcal{E}})_a\langle  \sqrt{-1} c({\rm d}x^a) c({\rm d}x^n) \sigma u_0, u_0 \rangle_{\delta} = \sin |\alpha|\, |\widehat{P}_{\mathcal{E}}|\, |u_0|_{\delta}^2.
\end{equation}
Denote by $C^{\infty}(X,S;Q,\chi)$ the space of all smooth sections of $S$ over $X$ that satisfy the boundary conditions \eqref{eq:BC-chi} and \eqref{eq:BC-Sigma}.
To summarize, putting \eqref{eq:term-of-compact-boundary}, \eqref{eq:term-of-noncompact-boundary}, \eqref{eq:mass1}, and \eqref{eq:mass2} into \eqref{eq:mass-spectral-estimate}, we obtain the following estimates.
For any $u\in C^{\infty}(X,S;Q,\chi)$,  
\begin{equation}\label{eq:Mr-spectral-estimate}
\begin{aligned}
\int_{M_r} &|\widetilde{\mathcal{D}}^{f} u|^2 dV \geq \int_{M_r} |\widetilde{\nabla}^{f}u|^{2} dV \\
& + \frac{n-1}{n} \int_{M_r} \left( \frac{n}{2(n-1)} (\mu - |J|) + f^2 - |{\rm d}f|  \right) |u|^2 dV \\
& + \int_{ \partial M_{r} \setminus ( S^{n-1}_{r,+} \cup \Sigma_{r}) }  \left(\frac{1}{2} \theta_{\partial M_{r} \setminus ( S^{n-1}_{r,+} \cup \Sigma_{r})}^{-} + \frac{n-1}{n}f \right) |u|^2 dA \\ 
& + \frac{1}{2} \int_{\Sigma_r} \left( H_{\Sigma_r}\mp \cos\alpha\, {\rm tr}_{\Sigma_r}k -\sin\alpha\, |k(\eta,\cdot)^{\top}| \right)|u|^2 dA\\
& - \frac{n-1}{n} \underbrace{ \int_{\Sigma_r} \langle f c(\eta^{\flat})\sigma u,u  \rangle dA}_{= 0} - \underbrace{\int_{S^{n-1}_{r,+}}\langle c(\nu^{\flat})\widetilde{\mathcal{D}} u + \widetilde{\nabla}_{\nu} u, u\rangle dA }_{ \to \frac{1}{4} (E_{\mathcal{E}} \pm \cos \alpha\, (P_{\mathcal{E}})_n -\sin |\alpha|\, |\widehat{P}_{\mathcal{E}}|) |u_0|_{\delta}^2,\, \text{as}\, r\to \infty } \\
&  - \frac{n-1}{n} \underbrace{ \int_{S^{n-1}_{r,+} }\langle f c(\nu^{\flat})\sigma u, u\rangle dA }_{ = 0 \, \text{for}\, r\gg 1}.
\end{aligned}
\end{equation}
We denote
\begin{equation}\label{eq:notation-Phi-Psi-Theta}
\begin{split}
     &\Phi^f =\frac{n}{2(n-1)}(\mu-|J|)+ f^2-|{\rm d}f|,\\
     &\Psi_{\partial X \setminus \Sigma }^f =\frac{1}{2} \theta_{\partial X \setminus \Sigma}^{-} +\frac{n-1}{n}f, \\
    &\Theta_{\Sigma} = \frac{1}{2} (H_{\Sigma}\mp \cos\alpha\, {\rm tr}_{\Sigma}k -\sin\alpha\, |k(\eta,\cdot)^{\top}|).
\end{split}
\end{equation}
Therefore, we are now in a position to characterize the ADM energy-momentum formula closely related to the spectral estimates of the Dirac-Witten operator with a Callias potential in the following Proposition \ref{pro:spectral-estimate} by letting $r\to \infty$ in \eqref{eq:Mr-spectral-estimate}.
\begin{proposition}\label{pro:spectral-estimate}
Let $(X,g,k)$ be a complete connected asymptotically flat spin initial data set with a non-compact boundary $\Sigma$ and compact boundary $\partial X\setminus \Sigma$. Assume that $M$ only has an asymptotically flat initial data end $\mathcal{E}$ with a non-compact boundary $\Sigma$. Let $\eta$ be the outward unit normal of $\partial X$. Let $f:X\to \mathbb{R}$ be a Lipschitz potential function with compact support such that ${\rm supp}(f)\cap \Sigma=\emptyset$. Let $u\in C^{\infty}(X,S;Q,\chi)$ which is asymptotic to a constant section $u_0$ satisfying \eqref{eq:BC-Q-bar} in the end $\mathcal{E}$. Then
\begin{equation}\label{eq:spectral-estimate}
\begin{aligned}
\|\widetilde{\mathcal{D}}^{f}u\|_{L^2(X,S)}^2  \geq & \|\widetilde{\nabla}^{f}u\|_{L^2(X,S)}^{2} + \frac{n-1}{n} \int_{X} \Phi^f |u|^2 dV + \int_{\partial X\setminus \Sigma} \Psi_{\partial X\setminus\Sigma}^f |u|^2 dA \\
& + \int_{\Sigma} \Theta_{\Sigma} |u|^2 dA - \frac{1}{4} (E_{\mathcal{E}} \pm \cos \alpha\, (P_{\mathcal{E}})_n -\sin |\alpha|\, |\widehat{P}_{\mathcal{E}}|) |u_0|_{\delta}^2.
\end{aligned}
\end{equation}
\end{proposition}
Another key ingredient in the sequel is the following proposition about a mixed boundary value problem for the Dirac-Witten operator with a Callias potential, compare with \cite[Proposition~4.7]{AdL20}, \cite[Proposition~6.2]{CW22} and \cite[Proposition~4.3]{CLZ23}.
\begin{proposition}\label{pro:existence-lemma}
Let $(X,g,k)$ be a complete connected asymptotically flat spin initial data set with a non-compact boundary $\Sigma$ and compact boundary $\partial X\setminus \Sigma$. Assume that $X$ only has an asymptotically flat initial data end $\mathcal{E}$ with a non-compact boundary $\Sigma$.
Assume that there exists a Lipschitz Callias potential function $f:X\to\mathbb{R}$ with compact support such that ${\rm supp}(f)\cap \Sigma=\emptyset$, $\Psi_{\partial X\setminus \Sigma}^f \geq 0$ on $\partial X\setminus \Sigma$, and $\Theta_{\Sigma} \geq 0$ on $\Sigma$, and write $\Phi^{f}= \Phi^{f}_{+} - \Phi^{f}_{-}$ in $X$ with $\Phi^{f}_{\pm}\geq 0$, where $\Psi_{\partial X\setminus \Sigma}^f$, $\Theta_{\Sigma}$, and $\Phi^{f}$ are defined as in \eqref{eq:notation-Phi-Psi-Theta}.
Let $X_1$ be a connected codimension zero submanifold that contains the end $\mathcal{E}$ and has a non-compact boundary $\Sigma$ and compact boundary $\partial X_1\setminus \Sigma$.
Suppose that ${\rm supp}(\Phi^{f}_{-})\subseteq X\setminus X_1$ and that the following holds:
\begin{itemize}
\item[(i)] $\Phi^{f}_{-}$ is small enough in $X\setminus X_1$,
\item[(ii)] $f^2/n^2 \leq \frac{w_0}{4}$ in $Y$, where $w_0$ is a weight constant of Poincar\'{e} inequality of $\widetilde{\nabla}$ as in \cite[Definition~8.2]{BC05} and $Y\subset X\setminus X_1$ is a compact set.
\end{itemize}
For any $u\in C^{\infty}(X,S)$ such that $\widetilde{\mathcal{D}}^{f} u\in L^2(X,S)$, 
there exists an unique $v\in {\rm H}_{{\rm loc}}^1(X,S;Q,\chi)$ solving the boundary value problem
\begin{equation}\label{eq:BVP}
\begin{cases}
\widetilde{\mathcal{D}}^{f} v = - \widetilde{\mathcal{D}}^{f} u \quad &\text{in}\ X,\\
Q \left.v\right|_{\Sigma} = \pm \left. v\right|_{\Sigma} \quad &\text{on}\ \Sigma,\\
\chi \left.v\right|_{\partial X\setminus \Sigma} = \left. v\right|_{\partial X\setminus \Sigma} \quad &\text{on}\ \partial X\setminus \Sigma,
\end{cases}
\end{equation}
where ${\rm H}_{{\rm loc}}^1(X,S;Q,\chi)$ denotes the Sobolev ${\rm H}^1$-space of locally sequare-integrable sections of $S$ over $X$ that satisfy the boundary conditions \eqref{eq:BC-chi} and \eqref{eq:BC-Sigma}.
\end{proposition}
\begin{proof}
We will prove this lemma by applying Theorem \ref{thm:appendix-thm} directly in the following.
First, note that $\chi$ and $Q$ are all self-adjoint involutions, so we may consider the projections
\begin{equation}
 P_{\pm Q}:=\frac{1}{2}({\rm Id}_{\left.S\right|_{\Sigma}} \pm Q): {\rm H}^{\frac{1}{2}}_{\rm loc}(\Sigma, \left.S\right|_{\Sigma}) \to B_{\pm Q}:={\rm H}^{\frac{1}{2}}_{\rm loc}(\Sigma,\left.S\right|_{\Sigma}^{\pm Q})
\end{equation} 
onto the $\pm 1$-eigenbundle $\left.S\right|_{\Sigma}^{\pm Q}$ of $Q$, and the projections 
\begin{equation}
 P_{\pm \chi}:=\frac{1}{2}({\rm Id}_{\left.S\right|_{\partial X\setminus \Sigma}} \pm \chi): {\rm H}^{\frac{1}{2}}(\partial X\setminus \Sigma, \left.S\right|_{\partial X\setminus \Sigma}) \to B_{\pm \chi}:={\rm H}^{\frac{1}{2}}(\partial X\setminus \Sigma, \left.S\right|_{\partial X\setminus \Sigma}^{\pm \chi})
\end{equation} 
onto the $\pm 1$-eigenbundle $\left.S\right|_{\partial X\setminus \Sigma}^{\pm \chi}$ of $\chi$.
Thus $({\rm Id_{\left.S\right|_{\Sigma}}}-P_{\pm Q})\mathscr{R}_1(v)=0$ if and only if $Q \left.v\right|_{\Sigma}=\pm \left.v\right|_{\Sigma}$, and $({\rm Id}_{\left. S\right|_{\partial X\setminus \Sigma}}-P_{+\chi})\mathscr{R}_2(v)=0$ if and only if $\chi \left.v\right|_{\partial X\setminus \Sigma}= \left.v\right|_{\partial X\setminus \Sigma}$. 
We remark that $B_{\pm Q}$ are nice elliptic boundary conditions for non-compact boundary $\Sigma$ (cf. \cite{GN14,BS18,AdL20}), and we note that $B_{\pm \chi}$ are also elliptic boundary conditions for compact boundary $\partial X\setminus \Sigma$ in the sense of B\"{a}r-Ballmann \cite{BB12,BB16}.
By the Green's formula, we have
\begin{equation}
\int_{X}\langle \widetilde{\mathcal{D}}^{f} u_1,u_2\rangle dV - \int_{X}\langle u_1, \widetilde{\mathcal{D}}^{f} u_2 \rangle dV = \int_{\partial X}\langle c(\eta^{\flat})u_1,u_2\rangle dA.
\end{equation}
for any $u_1,u_2\in {\rm H}_{\rm loc}^1(X,S)$. Thus we have $\langle c(\eta^{\flat}) u_1,u_2\rangle=0$ on $\partial X$ if $u_1,u_2\in {\rm H}^1_{\rm loc}(X,S;Q,\chi)$.
Since $B_{\pm \chi}$ and $B_{\pm Q}$ are self-adjoint boundary conditions, $\widetilde{\mathcal{D}}^{f}: {\rm dom}(\widetilde{\mathcal{D}}^{f})\subset {\rm H}_{\rm loc}^1(X,S;Q,\chi)$ $\to L^2(X,S)$ is essentially self-adjoint (cf. \cite{GN14,BB16,BS18,AdL20}).

Next, if $u \in {\rm dom}(\widetilde{\mathcal{D}}^f) \subset  {\rm H}^1_{\rm loc}(X,S;Q,\chi)$ with $\widetilde{\mathcal{D}}^{f}u=0$, by \eqref{eq:spectral-estimate} in Proposition \ref{pro:spectral-estimate} and noting that the last term in \eqref{eq:spectral-estimate} vanishes, we find  
\begin{equation}\label{eq:BVP-coercive}
\begin{aligned}
\|\widetilde{\mathcal{D}}^{f}u\|_{L^2(X,S)}^2 \geq & \|\widetilde{\nabla}^{f}u\|_{L^2(X,S)}^{2} + \frac{n-1}{n} \int_{X} \Phi^{f} |u|^2 dV \\
& +  \int_{\partial X \setminus \Sigma} \Psi_{\partial X\setminus \Sigma}^{f} |u|^2 dA +  \int_{\Sigma} \Theta_{\Sigma} |u|^2 dA \\
\geq & \|\widetilde{\nabla}^{f}u\|_{L^2(X,S)}^{2} - \int_{X} \Phi^{f}_{-} |u|^2 dV\\
\geq & \|\widetilde{\nabla} u\|_{L^2(Y,S)}^{2} - \left\|\tfrac{f}{n} u\right\|_{L^2(Y,S)}^2 - \int_{X\setminus X_1} \Phi^{f}_{-} |u|^2 dV\\
\geq & \frac{2}{3}\|\widetilde{\nabla} u\|_{L^2(Y,S)}^{2} - \left\|\tfrac{f}{n} u\right\|_{L^2(Y,S)}^2 + \frac{1}{3}\|\widetilde{\nabla} u\|_{L^2(Y,S)}^{2} - \int_{X\setminus X_1} \Phi^{f}_{-} |u|^2 dV\\
\geq & \frac{2}{3} \int_{Y} (w_0- \frac{f^2}{n^2})|u|^2 dV + \frac{1}{3} \int_{Y} w_0 |u|^2 dV - \int_{X\setminus X_1} \Phi^{f}_{-} |u|^2 dV \\
\geq &  \frac{2}{3} \int_{Y} (w_0-\frac{w_0}{4}) |u|^2 dV = \frac{w_0}{2}\|u\|_{L^2(Y,S)}^2,
\end{aligned}
\end{equation}
where the last second line of \eqref{eq:BVP-coercive} holds due to a weighted Poincar\'{e} inequality $\|\widetilde{\nabla} u\|_{L^2(Y,S)}^2$ $\geq w_0 \|u\|_{L^2(Y,S)}^2$ for a weight constant $w_0>0$, and the inequality in the last line of \eqref{eq:BVP-coercive} holds since $\Phi^{f}_{-}$ is small enough in $X\setminus X_1$ and we may assume that the compact subset $Y\subset X\setminus X_1$ satisfies ${\rm supp}(u)\cap Y\neq \emptyset$ and
\begin{equation}
\frac{1}{3} \int_{Y} w_0 |u|^2 dV \geq \int_{X\setminus X_1} \Phi^{f}_{-} |u|^2 dV,
\end{equation}
otherwise the term $\int_{X\setminus X_1}\Phi^{f}_{-}|u|^2dV$ in the last second line must vanish.
From \eqref{eq:BVP-coercive}, $u$ vanishes on $Y$. Moreover, from the first two lines in the previous estimate we have $\|\widetilde{\nabla}^{f}u\|_{L^2(X,S)}=0$. Theorefore, $\widetilde{\nabla}^{f}u=0$ on all of $X$. By \cite[Remark~A.2]{CLZ23}, $u=0$ everywhere and hence $\ker \widetilde{\mathcal{D}}^{f}=0$.

Finally, we let $K$ be a compact subset of $M$ such that ${\rm supp}(f)\subset M\setminus K$ and denote $C_{cc}^\infty(X,S)$ the space of all smooth sections of $S$ with compact support contained in the interior of X.
For any $u\in C_{cc}^\infty(X,S)$ with compact support $K_0$ in $X\setminus K$, we have
\begin{equation}
\begin{aligned}
\|\widetilde{\mathcal{D}}^{f}u\|_{L^2(X,S)}^{2} \geq & \int_{X}\left|\mathcal{D}u- \frac{{\rm tr}_{g}k}{2} \sigma u\right|^2 dV \\
=& \|\mathcal{D}u\|_{L^2(X,S)}^2 + \frac{({\rm tr}_{g}k)^2}{4} \|u\|_{L^2(X,S)}^2 \\
 & - \frac{{\rm tr}_{g}k}{2} \int_{X} (\langle\mathcal{D}u,\sigma u \rangle + \langle \sigma u, \mathcal{D} u \rangle ) dV \\
=& \|\mathcal{D}u\|_{L^2(K_0,S)}^2 + \frac{({\rm tr}_{g}k)^2}{4} \|u\|_{L^2(K_0,S)}^2 \\
 & + \frac{{\rm tr}_{g}k}{2} \int_{\partial X} \langle u, c(\nu^{\flat}) \sigma u \rangle   \\
\geq & C \|u\|_{L^2(K_0,S)}^2
\end{aligned}
\end{equation}
for some constant $C=\min_{K_0}\frac{({\rm tr}_{g}k)^2}{4}$, where the last inequality holds because $u$ vanishes on $\partial X$. Thus $\widetilde{\mathcal{D}}^{f}$ is coercive at infinity, and then $\widetilde{\mathcal{D}}^{f}$ is $(B_{Q},B_{\chi})$-coercive at infinity from Remark \ref{rem:appendix-remark-2}. 
By using Theorem \ref{thm:appendix-thm}, we find an unique $v\in {\rm H}_{{\rm loc}}^1(X,S;Q,\chi)$ satisfying \eqref{eq:BVP}. Hence our assertion follows.
\end{proof}

\section{Proofs of Theorem \ref{thm:flat--gener-PTM-arends} and Theorem \ref{thm:PMT-with-dimension-zero-submanifolds} }\label{sec:flat}
After establishing the ADM energy-momentum formula resting on spectral estimates and solving the mixed boundary value problem for the Dirac-Witten operator with a Callias potential in Section \ref{sec:pre}, we are now ready to prove Theorem \ref{thm:flat--gener-PTM-arends} and Theorem \ref{thm:PMT-with-dimension-zero-submanifolds}.

\subsection{Tilted spacetime positive mass theorem with arbitrary ends}
In this subsection, we prove Theorem \ref{thm:flat--gener-PTM-arends}.
\begin{proof}[Proof of Theorem \ref{thm:flat--gener-PTM-arends}]
We will prove the contrapositive statement of Theorem \ref{thm:flat--gener-PTM-arends}. Given a distance $\varrho>0$, we assume that there exist connected codimension zero submanifolds $M_1\subset M_{\varrho} \subseteq M$ such that we have the following (see Figure \ref{fig:Quantitative-shilding})
\begin{itemize}
     \item[(i)] ${\rm dist}_g(\partial M_1\setminus\Sigma, \partial M_{\varrho}\setminus\Sigma)>\varrho$,
     \item[(ii)] $M_{\varrho}\setminus M_1$ is relatively compact,
     \item[(iii)] $M_1$ contains an asymptotically flat initial data end $\mathcal{E}$ with a non-compact boundary $\Sigma$,
     \item[(iv)] $(M_{\varrho},g)$ is a complete spin manifold such that $\mu-|J|\geq 0$ holds on $M_{\varrho}$. 
\end{itemize}
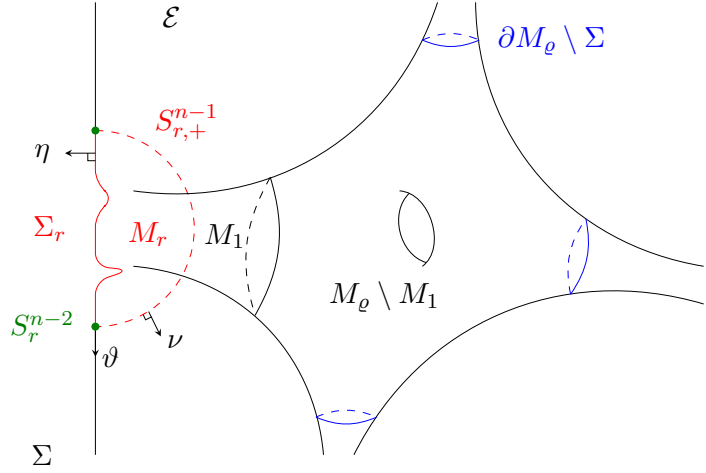
\begin{figure}
\begin{center}
\begin{tikzpicture}[scale=1]
\coordinate (A) at (0,3);
\coordinate (B) at (0,-3);
\coordinate (C) at (0,1.3);
\coordinate (D) at (0,-1.3);
\coordinate (E) at (0.5,0.5);
\coordinate (F) at (4.5,3);
\coordinate (G) at (0.5,-0.5);
\coordinate (H) at (3,-3);
\coordinate (I) at (5,3);
\coordinate (J) at (8,-0.5);
\coordinate (K) at (3.4,-3);
\coordinate (L) at (8,-1);
\coordinate (M) at (4,0.5);
\coordinate (N) at (4.3,-0.5);
\coordinate (O) at (4.13,0.47);
\coordinate (P) at (4.35,-0.45);
\draw[rounded corners=5,line width=0.2pt] (M) to [bend left=70] (N);
\draw[rounded corners=5,line width=0.2pt] (O) to [bend right=60] (P);
\draw[rounded corners=5,line width=0.2pt] (A) to (C);
\draw[rounded corners=5,line width=0.2pt] (B) to (D);
\draw[rounded corners=5,red,fill=white,line width=0.2pt] (C) to (0,0.6) to (0.25,0.4) to (0,0.2) to (0,-0.5) to (0.45,-0.55) to (0,-0.7) to (D);
\draw[dashed,red,line width=0.2pt] (C) arc [start angle=90, end angle=-90, x radius=1.3, y radius=1.3];
\draw[rounded corners=5,line width=0.2pt] (E) to [bend right=40] (F);
\draw[rounded corners=5,line width=0.2pt] (G) to [bend left=40]  (H);
\draw[rounded corners=5,line width=0.2pt] (I) to [bend right=40] (J);
\draw[rounded corners=5,line width=0.2pt] (K) to [bend left=40]  (L);
\draw[-stealth,line width=0.2pt] (0,1) to (-0.4,1);
\draw[line width=0.2pt] (-0.1,1) to (-0.1,0.9) to (0,0.9);
\draw[-stealth,line width=0.2pt] (D) to (0,-1.7);
\draw[-stealth,line width=0.2pt] (0.7,-1.1) to (0.86,-1.43);
\draw[line width=0.2pt] (0.628,-1.14) to (0.66,-1.21) to (0.74,-1.17);
\draw[dashed,rounded corners=5,line width=0.2pt] (2.3,0.68) to [bend right=23] (2.1,-1.15);
\draw[rounded corners=5,line width=0.2pt] (2.3,0.68) to [bend left=23] (2.1,-1.15);
\draw[rounded corners=5,blue,line width=0.2pt] (6.25,-0.87) to [bend right=23] (6.45,0.13);
\draw[dashed,rounded corners=5,blue,line width=0.2pt] (6.25,-0.87) to [bend left=23] (6.45,0.13);
\draw[rounded corners=5,blue,line width=0.2pt] (2.9,-2.5) to [bend right=23] (3.7,-2.5);
\draw[dashed,rounded corners=5,blue,line width=0.2pt] (2.9,-2.5) to [bend left=23] (3.7,-2.5);
\draw[rounded corners=5,blue,line width=0.2pt] (4.3,2.5) to [bend right=23] (5.05,2.5);
\draw[dashed,rounded corners=5,blue,line width=0.2pt] (4.3,2.5) to [bend left=23] (5.05,2.5);
\node at (1.2,1.4) {\(\textcolor{red}{S^{n-1}_{r,+}}\)};
\node at (-0.7,-3) {\(\Sigma\)};
\node at (1,2.8) {\(\mathcal{E}\)};
\node at (-0.7,-1.3) {\textcolor{green!50!black}{\(S^{n-2}_{r}\)}};
\node at (1.7,-0.1) {\(M_1\)};
\node at (3.8,-0.9) {\(M_{\varrho}\setminus M_1\)};
\node at (6,2.5)  {\(\textcolor{blue}{\partial M_{\varrho}\setminus \Sigma}\)};
\node at (0.7,-0.07)  {\(\textcolor{red}{M_r}\)};
\node at (-0.6,0)  {\(\textcolor{red}{\Sigma_r}\)};
\node at (-0.7,1)  {\(\eta\)};
\node at (0.2,-1.7)  {\(\vartheta\)};
\node at (1.05,-1.5) {\(\nu\)};
\node[circle,fill=green!50!black,minimum size=3pt,inner sep=0pt] at (C) {};
\node[circle,fill=green!50!black,minimum size=3pt,inner sep=0pt] at (D) {};
\end{tikzpicture}
\caption{$\mathcal{E}\subset M_1\subset M_{\varrho}\subseteq M$.}\label{fig:Quantitative-shilding}
\end{center}
\end{figure}
From (ii), since $\overline{M_{\varrho}\setminus M_1}$ is compact, $\mathcal{E}$ is the only asymptotically flat end with a non-compact boundary $\Sigma$ in $M_1$. From (i), there exists an $\varepsilon>0$ such that
\begin{equation}
   {\rm dist}_g(\partial M_1\setminus\Sigma, \partial M_{\varrho}\setminus\Sigma)\geq \omega:=\varrho + \varepsilon.  
\end{equation}
Let $x(p)={\rm dist}_g(p,\partial M_{\varrho}\setminus\Sigma)$. Then we construct a sequence of bounded Lipschitz functions $f_{\varrho,j}$ on $M_{\varrho}$ as in the process of proof of \cite[Theorem~1.4]{Liu23a}. For any $j$, let
\begin{equation}
f_{\varrho,j}(p)=\begin{cases}
\frac{\pi}{2\omega}\cot\left(\frac{\pi}{2\omega}x(p)+\frac{1}{j}\right) \quad &\text{if}\, x(p)\leq \frac{2\omega}{\pi} \left(\frac{\pi}{2}-\frac{1}{j}\right), \\
0 \quad &\text{otherwise.}
\end{cases}
\end{equation}
Observe that $f_{\varrho,j}\to \infty$ on $\partial M_{\varrho} \setminus\Sigma$ as $j\to \infty$. Moreover, note that $f_{\varrho,j}=0$ on $M_1$. We fix a compact subset $\Omega_{\varrho}$ in the interior of $M_{\varrho}$, such that for all sufficiently large $j$ we have
\begin{subequations}
\begin{align}
\frac{3}{2}f_{\varrho,j}^2-|{\rm d} f_{\varrho,j}|\geq 1\quad &\text{on}\  M_{\varrho} \setminus\Omega_{\varrho},\label{eq:characteristic-1}\\
f_{\varrho,j}^2-|{\rm d} f_{\varrho,j}|+\frac{\pi^2}{4 {\varrho}^2}\geq C_{\varepsilon,\varrho}\quad &\text{on}\  M_{\varrho}, \label{eq:characteristic-2}
\end{align}
\end{subequations}
where $C_{\varepsilon,\varrho}>0$ depends on $\varepsilon$ and $\varrho$. Fix a section $\varphi\in C^{\infty}(M_1,S)$ which is asymptotic to a non-zero constant section $u_0$ in $\mathcal{E}$ satisfying \eqref{eq:BC-Q-bar} and such that ${\rm supp}(\varphi)\subseteq M_1$ in advance. Since ${\rm supp}(\varphi)\cap {\rm supp}(f_{\varrho,j})=\emptyset$ for any $j$, we have
\begin{equation}
    \widetilde{\mathcal{D}}^{f_{\varrho,j}}\varphi=\widetilde{\mathcal{D}}\varphi \in L^2(M_1,S)\subset L^2(M_{\varrho},S).
\end{equation}
By the interior and tilted boundary dominant energy conditions from the hypotheses of Theorem \ref{thm:flat--gener-PTM-arends}, and construction of $f_{\varrho,j}$, we have 
\begin{equation}
\Psi_{\partial M_{\varrho}\setminus \Sigma}^{f_{\varrho,j}}\geq 0 \quad  \text{on}\quad \partial M_{\varrho} \setminus \Sigma, 
\end{equation}
\begin{equation}
\Theta_{\Sigma} \geq 0 \quad  \text{on}\quad  \Sigma,
\end{equation}
and
\begin{equation}
\Phi^{f_{\varrho,j}} = \frac{n}{2(n-1)}(\mu-|J|)+f_{\varrho,j}^2-|{\rm d}f_{\varrho,j}| \geq 
\begin{cases} 
0 \quad & \text{in}\quad  M_1, \\
-\frac{\pi^2}{4\varrho^2} \quad & \text{in}\quad  M_{\varrho}\setminus M_1.
\end{cases}
\end{equation}
We may find a connected compact codimension zero submanifold $M_{N} \subset M_{\varrho}\setminus M_1$ such that ${\rm dist}_g(p,M_1)< \frac{2\omega}{\pi}(\frac{\pi}{4}-\frac{1}{j})$ for all $p\in M_N$.
Then we can choose a constant $\varrho_0>0$ such that for all $\varrho>\varrho_0$, we have
\begin{itemize}
\item[(i)]
$\tfrac{f_{\varrho,j}^2}{n^2} \leq \frac{\pi^2}{4n^2\omega^2}\leq \frac{w_0}{4}$ in $M_N$,
\item[(ii)]
$\Phi^{f}_{-} \leq \frac{\pi^2}{4\varrho^2}$ and hence $\Phi^{f}_{-}$ is small enough in $M_{\varrho}\setminus M_1$.
\end{itemize}
Then by Proposition \ref{pro:existence-lemma} we find a section $v_{\varrho,j}\in {\rm H}_{{\rm loc}}^1(M_{\varrho},S;Q,\chi)$ such that
\begin{equation}
\begin{cases}
\widetilde{\mathcal{D}}^{f_{\varrho,j}} v_{\varrho,j} = - \widetilde{\mathcal{D}}^{f_{\varrho,j}} \varphi \quad &\text{in}\ M_{\varrho},\\
Q(\left.v_{\varrho,j}\right|_{\Sigma}) = \pm \left. v_{\varrho,j}\right|_{\Sigma} \quad &\text{on}\ \Sigma,\\
\chi(\left.v_{\varrho,j}\right|_{\partial M_{\varrho} \setminus \Sigma}) = \left. v_{\varrho,j}\right|_{\partial M_{\varrho} \setminus \Sigma} \quad &\text{on}\ \partial M_{\varrho}\setminus \Sigma,
\end{cases}
\end{equation}
for all $j$. Set $u_{\varrho,j}:=v_{\varrho,j}+\varphi$. Then $\widetilde{\mathcal{D}}^{f_{\varrho,j}}u_{\varrho,j}=0$ everywhere, and $u_{\varrho,j}=v_{\varrho,j}$ on $M_{\varrho} \setminus {\rm supp}(\varphi)$, and $u_{\varrho,j}$ is asymptotic to $\varphi$ in $\mathcal{E}$. From \eqref{eq:spectral-estimate} in Proposition \ref{pro:spectral-estimate}, we have
\begin{equation}\label{eq:key-pre-spectral}
\begin{aligned}
&- \int_{ \partial M_{\varrho} \setminus \Sigma }  \Psi_{ \partial M_{\varrho} \setminus \Sigma }^{f_{\varrho,j}} |u_{\varrho,j}|^2 dA \\
\geq & \int_{M_{\varrho}} |\widetilde{\nabla}^{f_{\varrho,j}} u_{\varrho,j}|^2 dV + \frac{n-1}{n} \int_{M_{\varrho}} \Phi^{f_{\varrho,j}} |u_{\varrho,j}|^2 dV\\
& +  \int_{\Sigma} \Theta_{\Sigma} |u_{\varrho,j}|^2 dA - \frac{1}{4} (E_{\mathcal{E}} \pm \cos \alpha\, (P_{\mathcal{E}})_n -\sin |\alpha|\, |\widehat{P}_{\mathcal{E}}|) |u_0|_{\delta}^2.
\end{aligned}
\end{equation}
We shall inspect the limit as $j\to \infty$, so we now proceed as in \cite{HKKZ23,Liu23a,Liu23b}. First note that applying $\left.f_{\varrho,j}\right|_{\Sigma}=0$ and an integration by parts yields 
\begin{equation}\label{eq:HKKZ-boundary-identity}
\begin{aligned}
\int_{\partial M_{\varrho}\setminus \Sigma}\langle c(\eta^{\flat}) u_{\varrho,j},f_{\varrho,j} \sigma u_{\varrho,j}\rangle dA &= \int_{M_{\varrho}} (\langle \mathcal{D}u_{\varrho,j}, f_{\varrho,j}\sigma u_{\varrho,j}\rangle - \langle u_{\varrho,j}, \mathcal{D}f_{\varrho,j}\sigma u_{\varrho,j} \rangle ) dV \\
&= -\int_{M_{\varrho}}(2f_{\varrho,j}^2|u_{\varrho,j}|^2+\langle u_{\varrho,j}, c({\rm d}f_{\varrho,j})\sigma u_{\varrho,j}\rangle) dV.
\end{aligned}
\end{equation}
Moreover, by \eqref{eq:characteristic-1}, \eqref{eq:HKKZ-boundary-identity}, and the fact that $f_{\varrho,j}$ blows-up on $\partial M_{\varrho}\setminus \Sigma$, we have
\begin{equation}\label{eq:noncom-control}
\begin{aligned}
\int_{M_{\varrho}\setminus\Omega_{\varrho}} \left(\frac{1}{2}f_{\varrho,j}^2+1\right)|u_{\varrho,j}|^2 dV & \leq \int_{M_{\varrho}\setminus \Omega_{\varrho}} (2f_{\varrho,j}^2-|{\rm d}f_{\varrho,j}|)|u_{\varrho,j}|^2 dV \\
&\leq \int_{\Omega_{\varrho}} (|{\rm d}f_{\varrho,j}|-2f_{\varrho,j}^2)|u_{\varrho,j}|^2 dV.
\end{aligned}
\end{equation}
It follows from \eqref{eq:noncom-control} that $\max_{\Omega_{\varrho}}|u_{\varrho,j}|\neq 0$, and thus we may assume that $\max_{\Omega_{\varrho}}|u_{\varrho,j}|=1$ by appropriate rescaling. Using \eqref{eq:characteristic-2} and \eqref{eq:noncom-control}, we have
\begin{equation}\label{eq:bound-1}
\int_{\Omega_{\varrho}} |u_{\varrho,j}|^2 dV + \int_{M_{\varrho} \setminus\Omega_{\varrho}}\left(\frac{1}{2}f_{\varrho,j}^2+1\right)|u_{\varrho,j}|^2 dV \leq \left(\frac{ \pi^2 }{4 \varrho^2} + 1\right) |\Omega_{\varrho}|.
\end{equation}
Let
\begin{equation}
\Upsilon_{\varrho,j}=\min_{\partial M_{\varrho} \setminus\Sigma}\left(\frac{n-1}{n} f_{\varrho,j} - \frac{1}{2} |\theta^{-}_{\partial M_{\varrho} \setminus\Sigma}| \right),
\end{equation}
then $\Upsilon_{\varrho,j} \to \infty$ as $j\to \infty$. Note that
\begin{equation}\label{eq:need}
\begin{aligned}
\int_{M_{\varrho}}|\widetilde{\nabla}^{f_{\varrho,j}}u_{\varrho,j}|^2 dV & \geq \int_{M_{\varrho}} |\widetilde{\nabla} u_{\varrho,j}|^2 dV - \int_{M_{\varrho}} \left|\frac{f_{\varrho,j}}{n}u_{\varrho,j}\right|^2 dV \\
&=\int_{M_{\varrho}} |\widetilde{\nabla} u_{\varrho,j}|^2 dV - \int_{M_{\varrho}\setminus M_1} \left|\frac{f_{\varrho,j}}{n}u_{\varrho,j}\right|^2 dV.
\end{aligned}
\end{equation}
By \eqref{eq:characteristic-2}, \eqref{eq:key-pre-spectral}, and \eqref{eq:need}, we have
\begin{equation}\label{eq:bound-2}
\begin{aligned}
& \int_{M_{\varrho}} |\widetilde{\nabla} u_{\varrho,j}|^2 dV + \int_{\partial M_{\varrho}\setminus\Sigma} \Upsilon_{\varrho,j} |u_{\varrho,j}|^2 dV \\
\leq & \int_{M_{\varrho}\setminus M_1} \left|\frac{f_{\varrho,j}}{n}u_{\varrho,j}\right|^2 dV +\int_{M_{\varrho}}|\widetilde{\nabla}^{f_{\varrho,j}}u_{\varrho,j}|^2 dV + \int_{\partial M_{\varrho}\setminus\Sigma} \Psi_{\partial M_{\varrho} \setminus\Sigma}^{f_{\varrho,j}} |u_{\varrho,j}|^2 dV \\
\leq & \int_{ M_{\varrho} \setminus M_1} \left|\frac{f_{\varrho,j}}{n}u_{\varrho,j}\right|^2 dV -\frac{n-1}{n} \int_{M_{\varrho}} \Phi^{f_{\varrho,j}} |u_{\varrho,j}|^2 dV - \int_{\Sigma} \Theta_{\Sigma} |u_{\varrho,j}|^2 dA \\
& + \frac{1}{4} (E_{\mathcal{E}} \pm \cos \alpha\, (P_{\mathcal{E}})_n -\sin |\alpha|\, |\widehat{P}_{\mathcal{E}}|) |u_0|_{\delta}^2 \\
\leq & \int_{ M_{\varrho}\setminus M_1} \left|\frac{f_{\varrho,j}}{n}u_{\varrho,j}\right|^2 dV -\frac{1}{2} \int_{M_{\varrho}} (\mu-|J|) |u_{\varrho,j}|^2 dV \\
& - \frac{n-1}{n}\int_{M_{\varrho} \setminus M_1} (C_{\varepsilon,\varrho}-\frac{\pi^2}{4 \varrho^2}) |u_{\varrho,j}|^2 dV \\
\leq & \int_{ M_{\varrho}\setminus M_1} \left|\frac{f_{\varrho,j}}{n}u_{\varrho,j}\right|^2 + \frac{(n-1)\pi^2}{4n\varrho^2} |u_{\varrho,j}|^2 dV 
\leq C_1
\end{aligned}
\end{equation}
for some constant $C_1$ independent of $j$. Then it follows from \eqref{eq:bound-1} and \eqref{eq:bound-2} that the sequence $u_{\varrho,j}$ is uniformly bounded in ${\rm H}_{{\rm loc}}^1(M_{\varrho},S)$. Thus $u_{\varrho,j}$ weakly subconverges to a function $u_{\varrho}\in {\rm H}_{{\rm loc}}^1(M_{\varrho},S)$ with strong convergence in ${\rm H}_{{\rm loc}}^s(M_{\varrho},S)$ for $s\in [\frac{1}{2},1)$. Since the trace map $\mathscr{R}_1: {\rm H}_{{\rm loc}}^{s}(M_{\varrho},S)\to {\rm H}^{s-\frac{1}{2}}_{\rm loc}(\Sigma,\left.S\right|_{\Sigma})$ and $\mathscr{R}_2: {\rm H}_{{\rm loc}}^{s}(M_{\varrho},S)\to {\rm H}^{s-\frac{1}{2}}(\partial M_{\varrho}\setminus \Sigma,\left.S\right|_{\partial M_{\varrho}\setminus \Sigma})$ are continuous (see Remark \ref{rem:appendix-remark-1}(1)), $u_{\varrho,j}$ converges subsequentially to $\mathscr{R}_1(u_{\varrho})$ in $L^2(\Sigma,\left.S\right|_{\Sigma})$ and $\mathscr{R}_2(u_{\varrho})$ in $L^2(\partial M_{\varrho}\setminus \Sigma,\left.S\right|_{\partial M_{\varrho}\setminus \Sigma})$ respectively. Moreover, $\mathscr{R}_2(u_{\varrho})=0$ on $\partial M_{\varrho}\setminus \Sigma$ since $\Upsilon_{\varrho,j}$ blows-up on $\partial M_{\varrho}\setminus \Sigma$. Then, from \eqref{eq:key-pre-spectral} we have the following estimates. For sufficiently large $\varrho>\varrho_0$,
\begin{equation}\label{eq:arbitrary-distance-spectral-estimate}
\begin{aligned}
&- \int_{ \partial M_{\varrho}\setminus \Sigma } \Psi_{ \partial M_{\varrho}\setminus \Sigma }^{f_{\varrho,j}} |u_{\varrho,j}|^2 dA \\
\geq & \int_{M_{\varrho}} |\widetilde{\nabla}^{f_{\varrho,j}} u_{\varrho,j}|^2 dV  + \frac{n-1}{n} \int_{M_{\varrho}} \Phi^{f_{\varrho,j}} |u_{\varrho,j}|^2 dV \\
& +  \int_{\Sigma} \Theta_{\Sigma} |u_{\varrho,j}|^2 dA - \frac{1}{4} (E_{\mathcal{E}} \pm \cos \alpha\, (P_{\mathcal{E}})_n -\sin |\alpha|\, |\widehat{P}_{\mathcal{E}}|) |u_0|_{\delta}^2 \\
\geq & \frac{w_0}{2} \|u_{\varrho,j}\|_{L^2(Y,S)}^2 \\
\geq & \frac{w_0}{2} \|u_{\varrho,j}\|_{L^2(Y \cap \Omega_{\varrho},S)}^2, \\
\end{aligned}
\end{equation}
where the last second inequality holds by using the same argument in \eqref{eq:BVP-coercive}.
Taking the limit as $j\to \infty$ in \eqref{eq:arbitrary-distance-spectral-estimate} while applying weak lower semi-continuity of the ${\rm H}_{{\rm loc}}^1$-norm, Fatou's lemma, and strong convergence in $L^2$, we have
\begin{equation}\label{eq:tocontr}
0\geq \frac{w_0}{2} \|u_{\varrho}\|_{L^2(Y \cap \Omega_{\varrho},S)}^2.
\end{equation}
However, the right-hand side in \eqref{eq:tocontr} is strictly positive since $\max_{\Omega_{\varrho}} |u_{\varrho}|=1$. Therefore, a contradiction is obtained, and we conclude the proof.
\end{proof}

\begin{proof}[{Proof of Theorem \ref{thm:chairesultwithbends}}]
This directly follows from Theorem \ref{thm:flat--gener-PTM-arends}.
\end{proof}

\subsection{Quantitative shielding theorem}
In this subsection, let us directly enter into the proof of Theorem \ref{thm:PMT-with-dimension-zero-submanifolds}.
\begin{proof}[Proof of Theorem \ref{thm:PMT-with-dimension-zero-submanifolds}]
The proof is similar to that of Theorem \ref{thm:flat--gener-PTM-arends}. We only point out the necessary modifications. Assume, by contradiction, that 
\begin{equation}\label{eq:contraenergy}
E_{\mathcal{E}} \pm \cos\alpha \, (P_{\mathcal{E}})_n  < \sin \alpha\, |\widehat{P}_{\mathcal{E}}|,\quad \alpha\in \left[0, \frac{\pi}{2}\right]. 
\end{equation}
Since $\overline{U_0\setminus \mathcal{E}}$ is compact, $\mathcal{E}$ is the only asymptotically flat end with a non-compact boundary $\Sigma$ in $U_0$.
From the hypothesis (4) of the theorem, there exists an $\varepsilon>0$ such that
\begin{equation}
{\rm dist}_{g}(\partial \overline{U}_1\setminus\Sigma,\partial \overline{U}_0\setminus\Sigma) \geq \omega:= \sqrt{\frac{d_0}{\mathcal{Q}}} + \varepsilon.     
\end{equation}
Let $x(p)={\rm dist}_g(p,\partial \overline{U}_0\setminus \Sigma)$. Similar to the proof of Theorem \ref{thm:flat--gener-PTM-arends}, for any $j$ we define a sequence of bounded Lipschitz functions $f_j$ on $\overline{U}_0$ in the following
\begin{equation}
f_j(p)=\begin{cases}
\frac{\pi}{2\omega}\cot\left(\frac{\pi}{2\omega}x(p)+\frac{1}{j}\right) \quad &\text{if}\, r(p)\leq \frac{2\omega}{\pi} \left(\frac{\pi}{2}-\frac{1}{j}\right), \\
0 \quad &\text{otherwise.}
\end{cases}
\end{equation} 
Observe that $f_j\to \infty$ on $\partial \overline{U}_0\setminus \Sigma$ as $j\to \infty$. Moreover, note that $f_j=0$ on $\overline{U}_1$. We fix a compact subset $\Omega\subset U_0$ such that for all sufficiently large $j$ we have
\begin{subequations}
\begin{align}
\frac{3}{2}f_j^2-|{\rm d} f_j|\geq 1\quad &\text{on}\ \overline{U}_0 \setminus\Omega,\label{eq:AE-characteristic-1}\\
f_j^2-|{\rm d} f_j|+\frac{\mathcal{Q}\pi^2}{4d_0}\geq C_{\varepsilon,n,\mathcal{Q}} \quad &\text{on}\ \overline{U}_0, \label{eq:AE-characteristic-2}
\end{align}
\end{subequations}
where $C_{\varepsilon,n,\mathcal{Q}}>0$ depends on $\varepsilon, n$, and $\mathcal{Q}$. As in the proof of Theorem \ref{thm:flat--gener-PTM-arends}, for any $j$ we obtain an asymptotically constant $u_j$ which is asymptotically to $\varphi$ such that $\widetilde{\mathcal{D}}^{f_j}u_j=0$.
By using \eqref{eq:spectral-estimate} of Proposition \ref{pro:spectral-estimate}, we have
\begin{equation}\label{eq:AE-pre-spectral}
\begin{aligned}
&- \int_{ \partial \overline{U}_{0} \setminus \Sigma }  \Psi_{ \partial \overline{U}_{0} \setminus \Sigma }^{f_j} |u_j|^2 dA \\
\geq & \int_{U_0} |\widetilde{\nabla}^{f_j} u_j|^2 dV + \frac{n-1}{n} \int_{U_0} \Phi^{f_j} |u_j|^2 dV\\
&  +  \int_{\Sigma} \Theta_{\Sigma} |u_j|^2 dA - \frac{1}{4} (E_{\mathcal{E}} \pm \cos \alpha\, (P_{\mathcal{E}})_n -\sin |\alpha|\, |\widehat{P}_{\mathcal{E}}|) |u_0|_{\delta}^2.
\end{aligned}
\end{equation}
Using the same discussion as in the proof of Theorem \ref{thm:flat--gener-PTM-arends}, we find that the sequence $u_j$ converges strongly to a function $u$ in ${\rm H}_{{\rm loc}}^s(\overline{U}_0,S)$ for $s\in [\frac{1}{2},1)$ and $\mathscr{R}_2(u)=0$ on $\partial \overline{U}_{0} \setminus \Sigma$. Then by taking into account \eqref{eq:AE-pre-spectral}, the interior and boundary dominant energy conditions in items (1) to (3) from the statements of Theorem \ref{thm:PMT-with-dimension-zero-submanifolds}, and the inequalities \eqref{eq:contraenergy}, \eqref{eq:AE-characteristic-2} yield
\begin{equation}\label{eq:AE-quantitative-neck-estimate}
\begin{aligned}
- \int_{ \partial U_{0} \setminus \Sigma }  \Psi_{ \partial \overline{U}_{0} \setminus \Sigma }^{f_j} |u_j|^2 dA \geq &\frac{n-1}{n} \int_{U_0\setminus \overline{U}_1} \left( \frac{\mathcal{Q} \pi^2 }{4 d_0} + f_j^2-|{\rm d}f_j|\right) |u_j|^2 dV \\
& + \frac{1}{2} \int_{\overline{U}_1} (\mu-|J|)|u_j|^2  dV \\
\geq & \frac{n-1}{n} \int_{(U_0\setminus \overline{U}_1)\cap \Omega } C_{\varepsilon,n,\mathcal{Q}} |u_j|^2 dV \\
& + \frac{1}{2} \int_{\overline{U}_1\cap \Omega } (\mu-|J|)|u_j|^2  dV.
\end{aligned}
\end{equation}
Furthermore, taking the limit as $j\to \infty$ in \eqref{eq:AE-quantitative-neck-estimate} with weak lower semi-continuity of the ${\rm H}_{{\rm loc}}^1$-norm, Fatou's lemma, and strong convergence in $L^2$, we obtain
\begin{equation}
\begin{aligned}
0 \geq \frac{n-1}{n} \int_{(U_0\setminus \overline{U}_1)\cap \Omega } C_{\varepsilon,n,\mathcal{Q}} |u|^2 dV +  \frac{1}{2} \int_{ \overline{U}_1\cap \Omega } (\mu-|J|)|u|^2  dV.
\end{aligned}
\end{equation}
Therefore, we obtain a contradiction since $\max_{\Omega} |u|=1$. It then follows that
\begin{equation}
E_{\mathcal{E}} \pm \cos\alpha \, (P_{\mathcal{E}})_n  \geq \sin \alpha\, |\widehat{P}_{\mathcal{E}}|, \quad \text{here}\  \alpha\in [0, \frac{\pi}{2}].
\end{equation}
This finishes the proof.
\end{proof}


\appendix

\section{Boundary value problems}\label{sec:appendix}
In this section, we discuss a mixed boundary value problem associated with the Dirac-Witten operator with a Callias potential on a complete connected asymptotically flat spin initial data set $X$ with a non-compact boundary $\Sigma$ and compact boundary $\partial X\setminus \Sigma$ (compare with \cite{BB12,GN14}). 

\begin{definition}[compare {\cite[Definition~8.2]{BB12},\cite[Definition~6.1]{GN14}}]
We say that $\widetilde{\mathcal{D}}^{f}$ is \textit{coercive at infinity} if there is a compact subset $K\subset M$ and a constant $c>0$ such that
\begin{equation}
\|\widetilde{\mathcal{D}}^{f} u\|_{L^2(X,S)} \geq c\|u\|_{L^2(X, S)}
\end{equation}
for all $u\in C_{cc}^{\infty}(M\setminus K, S)$.
\end{definition}

\begin{definition}[{\cite[Definition~4.17]{GN14}}]
We say that $\widetilde{\mathcal{D}}^{f}$ is \textit{$({\rm dom}(\widetilde{\mathcal{D}}^{f}))$-coercive at infinity} if there is a constant $c>0$ such that
\begin{equation}
\forall u \in {\rm dom}(\widetilde{\mathcal{D}}^{f}) \cap (\ker(\widetilde{\mathcal{D}}^{f}))^{\bot}: \| \widetilde{\mathcal{D}}^{f} u\|_{L^2(X,S)} \geq c\|u\|_{L^2(X,S)},
\end{equation}
where $\bot$ denotes the orthogonal complement in $L^2(X,S)$.
\end{definition}
\begin{remark}\label{rem:appendix-remark-1}
\begin{itemize}
\item[(1)] From the trace theorem (cf. \cite[Theorem~3.7]{GN14}, \cite[Facts~5.4(v)]{BB12}), for all $s\in\mathbb{R}$ with $s>\frac{1}{2}$, the two trace maps $\mathscr{R}_1:C_c^{\infty}(X,S)\to C_c^{\infty}(\Sigma,\left.S\right|_{\Sigma})$, $u\mapsto \left.u\right|_{\Sigma}$ and $\mathscr{R}_2:C_c^{\infty}(X,S)\to C^{\infty}(\partial X\setminus \Sigma,\left.S\right|_{\partial X\setminus \Sigma}), u\mapsto \left.u\right|_{\partial X\setminus \Sigma}$ can be extended to the continuous linear operators $\mathscr{R}_1:{\rm H}_{\rm loc}^{s}(X,S)\to {\rm H}^{s-\frac{1}{2}}_{\rm loc}(\Sigma,\left.S\right|_{\Sigma})$ and $\mathscr{R}_2:{\rm H}_{\rm loc}^{s}(X,S)\to {\rm H}^{s-\frac{1}{2}}(\partial X\setminus \Sigma,\left.S\right|_{\partial X\setminus \Sigma})$ respectively.
\item[(2)] 
The domain ${\rm dom}(\widetilde{\mathcal{D}}^{f})$ of a closed extension of $\widetilde{\mathcal{D}}^{f}_{cc}$ with $B_1:=\mathscr{R}_1({\rm dom}(\widetilde{\mathcal{D}}^{f}))\subset {\rm H}^{-\frac{1}{2}}_{\rm loc}(\Sigma,\left.S\right|_{\Sigma})$ and $B_2:=\mathscr{R}_2({\rm dom}(\widetilde{\mathcal{D}}^{f}))\subset {\rm H}^{-\frac{1}{2}}(\partial X\setminus \Sigma,\left.S\right|_{\partial X\setminus \Sigma})$ equals 
\begin{equation}
{\rm dom}(\widetilde{\mathcal{D}}^{f}_{B_1,B_2}):=\{u\in {\rm dom}(\widetilde{\mathcal{D}}^{f}_{\rm max})|\mathscr{R}_1 u\in B_1,\mathscr{R}_2 u\in B_2\},
\end{equation}
where $\widetilde{\mathcal{D}}^{f}_{cc}$ denotes the operator $\widetilde{\mathcal{D}}^{f}$ with domain ${\rm dom}(\widetilde{\mathcal{D}}^{f}_{cc})=C_{cc}^{\infty}(X,S)$, and $\widetilde{\mathcal{D}}^{f}_{\rm max}$ denotes the maximal extension of $\widetilde{\mathcal{D}}^{f}$.
Moreover, $B_1\subset \check{\mathscr{R}}, B_2\subset {\rm \check{H}(A)}$ are closed subspaces of $\check{\mathscr{R}}, {\rm \check{H}(A)}$ (see \cite[Lemma~4.8]{GN14}, \cite[(36)]{BB12} for their precise definitions) respectively.
Conversely, for any closed subspaces $B_1\subset \check{\mathscr{R}}$, $B_2\subset {\rm \check{H}(A)}$, the operator $\widetilde{\mathcal{D}}^{f}_{B_1,B_2}:{\rm dom}(\widetilde{\mathcal{D}}^{f}_{B_1,B_2})\to L^2(X,S)$ is a closed extension of $\widetilde{\mathcal{D}}^{f}_{cc}$, cf. \cite[Lemma~4.13]{GN14}, \cite[Proposition~7.2]{BB12}.
\end{itemize}
\end{remark}
\begin{definition}
A closed subspace $B_1$ of $\check{\mathscr{R}}$ is called a boundary condition for non-compact boundary $\Sigma$, and a closed subspace $B_2$ of ${\rm \check{H}(A)}$ is called a boundary condition for compact boundary $\partial X\setminus\Sigma$.
 \end{definition}
For simplicity, we will use the short notation --- $\widetilde{\mathcal{D}}^f$ is $(B_1,B_2)$-coercive at infinity when $\widetilde{\mathcal{D}}^{f}_{B_1,B_2}$ is $({\rm dom}(\widetilde{\mathcal{D}}^{f}_{B_1,B_2}))$-coercive at infinity. 
\begin{remark}\label{rem:appendix-remark-2}
If $\widetilde{\mathcal{D}}^f$ is coecive at infinity and $B_1\subset {\rm H}^{\frac{1}{2}}_{\rm loc}(\Sigma,\left.S\right|_{\Sigma}), B_2\subset {\rm H}^{\frac{1}{2}}(\partial X\setminus \Sigma,\left.S\right|_{\partial X\setminus \Sigma})$, then $\widetilde{\mathcal{D}}^f$ is $(B_1,B_2)$-coercive at infinity, compare \cite[Corollary~6.4]{GN14}.
\end{remark}
\begin{theorem}\label{thm:appendix-thm}  
Let $B_1,B_2$ be two boundary conditions such that $B_1\subset {\rm H}^{\frac{1}{2}}_{\rm loc}(\Sigma,\left.S\right|_{\Sigma})$ and $B_2\subset {\rm H}^{\frac{1}{2}}(\partial X\setminus \Sigma,\left.S\right|_{\partial X\setminus \Sigma})$ respectively. 
We assume that $\widetilde{\mathcal{D}}^{f}: {\rm dom}(\widetilde{\mathcal{D}}^{f}_{B_1,B_2}) \subset L^2(X,S) \to L^2(X,S)$ is $(B_1,B_2)$-coercive at infinity. Let $P_{B_1}: {\rm H}^{\frac{1}{2}}_{\rm loc}(\Sigma,\left.S\right|_{\Sigma})\to B_1$ and $P_{B_2}: {\rm H}^{\frac{1}{2}}(\partial X\setminus\Sigma,\left.S\right|_{\partial X\setminus \Sigma})\to B_2$ be two projections. Then for all $u \in L^2(X,S)$ such that $u \in (\ker(\widetilde{\mathcal{D}}^{f}_{B_1,B_2})^{*})^{\bot}$, the boundary value problem
\begin{equation}
\begin{cases}
\widetilde{\mathcal{D}}^{f}v=u \quad & \text{on}\ X,\\
({\rm Id}_{\left.S\right|_{\Sigma}}-P_{B_1})\mathscr{R}_1(v) = 0 \quad & \text{on}\ \Sigma,\\
({\rm Id}_{\left.S\right|_{\partial X\setminus \Sigma}}-P_{B_2})\mathscr{R}_2(v) = 0 \quad & \text{on}\ \partial X\setminus \Sigma
\end{cases}
\end{equation}
has a solution $v \in {\rm H}_{\rm loc}^1(X,S)$ that is unique up to elements of the kernel $\ker(\widetilde{\mathcal{D}}^{f}_{B_1,B_2})$.
\end{theorem}
\begin{proof}
It essentially follows the proof of Corollary~4.19 in \cite{GN14} except additional arguments about a boundary condition for compact boundary from \cite{BB12}.
From \cite[Lemma~4.6]{GN14} and \cite[Theorem~6.7(iii)]{BB12}, $B_1\subset {\rm H}^{\frac{1}{2}}_{\rm loc}(\Sigma,\left.S\right|_{\Sigma})$ and $B_2\subset {\rm H}^{\frac{1}{2}}(\partial X\setminus \Sigma,\left.S\right|_{\partial X\setminus \Sigma})$ imply ${\rm dom}(\widetilde{\mathcal{D}}^{f}_{B_1,B_2})\subset {\rm H}^1_{\rm loc}(X,S)$. 
Since $\widetilde{\mathcal{D}}^{f}$ is $(B_1,B_2)$-coercive at infinity, its range is closed \cite[Lemma~4.18]{GN14}.
By the closed range theorem, $u \in (\ker(\widetilde{\mathcal{D}}^{f}_{B_1,B_2})^*)^{\bot}={\rm ran}(\widetilde{\mathcal{D}}^{f}_{B_1,B_2})$. Therefore, there is $v\in {\rm dom}(\widetilde{\mathcal{D}}^{f}_{B_1,B_2}) \subset {\rm H}_{\rm loc}^1(X,S)$ with $\widetilde{\mathcal{D}}^{f}v=u$, $({\rm Id}_{\left.S\right|_{\Sigma}}-P_{B_1})\mathscr{R}_1(v)=0$ and $({\rm Id}_{\left.S\right|_{\partial X\setminus \Sigma}}-P_{B_2})\mathscr{R}_2(v)=0$. 
\end{proof}

\vspace{.3cm}

\bibliographystyle{alpha}
\bibliography{tilted-PMT-arxiv}

\begin{thebibliography}{BHM{\etalchar{+}}15}

\bibitem[ABdL16]{ABdL16} S. Almaraz, E. Barbosa, and L. L. de Lima. {\sl A positive mass theorem for asymptotically flat manifolds with a non-compact boundary}. Commun. Anal. Geom. 24.4 (2016), pp. 673--715. DOI: \href{https://dx.doi.org/10.4310/CAG.2016.v24.n4.a1}{10.4310/CAG.2016.v24.n4.a1}.


\bibitem[AdL20]{AdL20} S. Almaraz and L. L. de Lima. {\sl The mass of an symptotically hyperbolic manifold with a non-compact boundary}. Ann. Henri Poincar\'{e} 21.11 (2020), pp. 3727--3756. DOI: \href{https://doi.org/10.1007/s00023-020-00954-w}{10.1007/s00023-020-00954-w}.


\bibitem[AdL22]{AdL22} S. Almaraz and L. L. de Lima. {\sl Rigidity of non-compact static domains in hyperbolic space via positive mass theorems}. Preprint. 2022. \href{https://arxiv.org/abs/2206.09768}{arXiv:2206.09768 [math.DG]}.


\bibitem[AdLM21]{AdLM21} S. Almaraz, L. L. de Lima, and L. Mari. {\sl Spacetime positive mass theorems for initial data sets with noncompact boundary}. Int. Math. Res. Not. 4 (2021), pp. 2783--2841. DOI: \href{https://doi.org/10.1093/imrn/rnaa226}{10.1093/imrn/rnaa226}.


\bibitem[Bar86]{Bar86} R. Bartnik. {\sl The mass of an asymptotically flat manifold}. Comm. Pure Appl. Math. 39.5 (1986), pp. 661--693. DOI: \href{ https://doi.org/10.1002/cpa.3160390505}{10.1002/cpa.3160390505}.


\bibitem[BB12]{BB12} C. B\"{a}r and W. Ballmann. {\sl Boundary value problems for elliptic differential operators of first order}. Surv. Differ. Geom. 17, Int. Press, Boston, MA (2012), pp. 1--78. DOI: \href{https://dx.doi.org/10.4310/SDG.2012.v17.n1.a1}{10.4310/SDG.2012.v17.n1.a1}.


\bibitem[BB16]{BB16} C. B\"{a}r and W. Ballmann. {\sl Guide to elliptic boundary value problems for Dirac-type operators}. Arbeitstagung Bonn 2013. Vol. 319. Progr. Math. Birkh\"{a}user/Spinger, Cham (2016), pp. 43--80. DOI: \href{https://doi.org/10.1007/978-3-319-43648-7_3}{10.1007/978-3-319-43648-7${}_{-}$3}.


\bibitem[BC03]{BC03} R. A. Bartnik and P. T. Chru\'{s}ciel. {\sl Boundary value problems for Dirac-type equations, with applications}. Preprint. 2003. \href{https://arxiv.org/abs/math/0307278}{arXiv:math/0307278 [math.DG]}.


\bibitem[BC05]{BC05} R. A. Bartnik and P. T. Chru\'{s}ciel. {\sl Boundary value problems for Dirac-type equations}. J. Reine Angew. Math. 279 (2005), pp. 13--73. DOI: \href{https://doi.org/10.1515/crll.2005.2005.579.13}{10.1515/crll.2005.2005.579.13}.


\bibitem[BdL23]{BdL23} R. M. Batista and L. L. de Lima. {\sl A harmonic level proof of a positive mass theorem}. Preprint. 2023. \href{https://arxiv.org/abs/2306.09097}{arXiv:2306.09097 [math.DG]}.


\bibitem[BKKS22]{BKKS22} H. Bray, D. Kazaras, M. Khuri, and D. Stern. {\sl Harmonic functions and the mass of $3$-dimensional asymptotically flat Riemannian manifolds}. J. Geom. Anal. 32(6):184, 2022. DOI: \href{https://doi.org/10.1007/s12220-022-00924-0}{10.1007/s12220-022-00924-0}.


\bibitem[BS18]{BS18} M. Braverman and P. Shi. {\sl The index of a local boundary value problem for strongly Callias-type operators}. Arnold Math. J. 5.1 (2019), pp. 79--96. DOI: \href{https://doi.org/10.1007/s40598-019-00110-1}{10.1007/s40598-019-00110-1}.


\bibitem[Cec20]{Cec20} S. Cecchini. {\sl A long neck principle for Riemannian spin manifolds with positive scalar curvature}. Geom. Funct. Anal. 20.5, (2020), pp. 1183--1223. DOI: \href{https://doi.org/10.1007/s00039-020-00545-1}{10.1007/s00039-020-00545-1}.


\bibitem[CH03]{CH03} P. T. Chru\'{s}ciel and M. Herzlich. {\sl The mass of asymptotically hyperbolic Riemannian manifolds}. Pacific J. Math. 212 (2003), pp. 231--264. DOI: \href{http://dx.doi.org/10.2140/pjm.2003.212.231}{10.2140/pjm.2003.212.231}.


\bibitem[Cha18]{Cha18} X. Chai. {\sl Positive mass theorem and free boundary minimal surfaces}. Preprint. 2018. \href{https://arxiv.org/abs/1811.06254}{arXiv:1811.06254 [math.DG]}.


\bibitem[Cha23]{Cha23} X. Chai. {\sl A tilted spacetime positive mass theorem}. Preprint. 2023. \href{https://arxiv.org/abs/2304.05208}{arXiv:2304.05208 [math.DG]}.


\bibitem[CLZ23]{CLZ23} S. Cecchini, M. Lesourd, and R. Zeidler. {\sl Positive mass theorems for spin initial data sets with arbitrary ends and dominant energy shields}. Preprint. 2023. \href{https://arxiv.org/abs/2307.05277}{arXiv:2307.05277 [math.DG]}.


\bibitem[CW22]{CW22} X. Chai and X. Wan. {\sl The mass of an asymptotically hyperbolic ends and distance estimates}. J. Math. Phys. 63.12 (2022), Paper No. 122502, 18. DOI: \href{https://doi.org/10.1063/5.0121452}{10.1063/5.0121452}.


\bibitem[CZ21a]{CZ21a} S. Cecchini and R. Zeidler. {\sl Scalar and mean curvature comparison via the Dirac operator}. Preprint. 2021. \href{https://arxiv.org/abs/2103.06833}{arXiv:2103.06833 [math.DG]}. To appear in Geom. Topol.


\bibitem[CZ21b]{CZ21b} S. Cecchini and R. Zeidler. {\sl Positive mass theorems and distance estimates in the spin setting}. Preprint. 2021. \href{https://arxiv.org/abs/2108.11972}{arXiv:2108.11972 [math.DG]}. To appear in Trans. Amer. Math. Soc.


\bibitem[EHLS16]{EHLS16} M. Eichmair, L.-H. Huang, D. Lee, and R. Schoen. {\sl The spacetime positive mass theorem in dimensions less than eight}. J. Eur. Math. Soc. 18.1 (2016), pp. 83--121. DOI: \href{https://doi.org/10.4171/jems/584}{10.4171/jems/584}.


\bibitem[GN14]{GN14} N. Grosse and R. Nakad. {\sl Boundary value problems for noncompact boundaries of ${\rm Spin^c}$ manifolds and spectral estimate}. Proc. Lond. Math. Soc. 109.4 (2014), pp. 946--974. DOI: \href{https://doi.org/10.1112/plms/pdu026}{10.1112/plms/pdu026}.


\bibitem[Gro18]{Gro18} M. Gromov. {\sl Metric inequality with scalar curvature}. Geom. Funct. Anal. 28.3 (2018), pp. 645--726. DOI: \href{https://doi.org/10.1007/s00039-018-0453-z}{10.1007/s00039-018-0453-z}.


\bibitem[Gro19]{Gro19} M. Gromov. {\sl Four lectures on scalar curvature}. Preprint. 2019. \href{https://arxiv.org/abs/1908.10612}{arXiv:1908.10612 [math.DG]}.


\bibitem[HKKZ23]{HKKZ23} S. Hirsch, D. Kazaras, M. Khuri, and Y. Zhang. {\sl Spectral torical band inequalities and generalizations of the Schoen-Yau black hole existence theorem}. Preprint. 2023. \href{https://arxiv.org/abs/2301.08270}{arXiv:2301.08270 [math.DG]}.


\bibitem[Lee19]{Lee19} D. A. Lee. {\sl Geometric Relativity}. Vol.201. Graduate Studies in Mathematics. American Mathematical Society, Providence, PI, 2019. DOI: \href{https://doi.org/10.1365/s13291-022-00245-9}{10.1365/s13291-022-00245-9}.


\bibitem[Liu23a]{Liu23a} D. Liu. {\sl A note on the long neck principle and spectral inequality of geodesic collar neighborhood}. Preprint. 2023. \href{https://arxiv.org/abs/2303.15333}{arXiv:2303.15333 [math.DG]}.


\bibitem[Liu23b]{Liu23b} D. Liu. {\sl On the long neck principle and width estimates for initial data sets}. Preprint. 2023. \href{https://arxiv.org/abs/2307.07278}{arXiv:2307.07278 [math.DG]}.


\bibitem[Liu+]{Liu+} D. Liu. {\sl A level set proof of a spacetime positive mass theorem with arbitrary ends and a non-compact boundary}. Preprint, in preparation. 


\bibitem[LLU22]{LLU22} D. A. Lee, M. Lesourd, and R. Unger. {\sl  Density and positive mass theorems for incomplete manifolds}. Preprint. 2022. \href{https://arxiv.org/abs/2201.01328}{arXiv:2201.01328 [math.DG]}. 


\bibitem[Loh16]{Loh16} J. Lohkamp. {\sl The higher dimensional positive mass theorem I}. Preprint. 2016. \href{https://arxiv.org/abs/math/0608795v2}{arXiv:math/0608795v2 [math.DG]}.


\bibitem[Loh17]{Loh17} J. Lohkamp. {\sl The higher dimensional positive mass theorem II}. Preprint. 2017. \href{https://arxiv.org/abs/1612.07505v2}{arXiv:1612.07505v2 [math.DG]}. 


\bibitem[LUY20]{LUY20} M. Lesourd, R. Unger, and S. T. Yau. {\sl  Positive scalar curvature on noncompact manifolds and the Liouville Theorem}. Preprint. 2020. \href{https://arxiv.org/abs/2009.12618}{arXiv:2009.12618 [math.DG]}. 


\bibitem[LUY21]{LUY21} M. Lesourd, R. Unger, and S. T. Yau. {\sl  The positive mass theorem with arbitrary ends}. Preprint. 2021. \href{https://arxiv.org/abs/2103.02744}{arXiv:2103.02744 [math.DG]}. To appear in J. Diff. Geom. 


\bibitem[Mia02]{Mia02} P. Miao. {\sl Positive mass theorem on manifolds admitting corners along a hypersurface}. Adv. Theor. Math. Phys. 6.6 (2002), pp. 1163--1182. DOI: \href{https://dx.doi.org/10.4310/ATMP.2002.v6.n6.a4}{10.4310/ATMP.2002.v6.n6.a4}. 


\bibitem[PT82]{PT82} T. Parker and C. H. Taubes. {\sl On Witten's proof of the positive energy theorem}. Comm. Math. Phys. 84.2 (1982), pp. 223--238. DOI: \href{https://doi.org/10.1007/BF01208569}{10.1007/BF01208569}.


\bibitem[Sch89]{Sch89} R. Schoen. {\sl Variational theory for the total scalar curvature functional for Riemannian metrics and related topics}. Topics in Calculus of Variations (Montecatini Terme, 1987), Lecture Notes in Math. 1365, Springer, Berlin, 1989, pp. 120--154. DOI: \href{https://doi.org/10.1007/BFb00891}{10.1007/BFb00891}.


\bibitem[SY79a]{SY79a} R. Schoen and S. T. Yau. {\sl On the proof of the positive mass conjecture in general relativity}. Comm. Math. Phys. 65.1 (1979), pp. 45--76. DOI: \href{https://doi.org/10.1007/BF01940959}{10.1007/BF01940959}.


\bibitem[SY79b]{SY79b} R. Schoen and S. T. Yau. {\sl Complete manifolds with nonnegative scalar curvature and the positive action conjecture in general relativity}. Proc. Nat. Acad. Sci. U.S.A. 76.3 (1979), pp. 1024--1025. DOI: \href{https://doi.org/10.1073/pnas.76.3.1024}{10.1073/pnas.76.3.1024}.


\bibitem[SY81a]{SY81a} R. Schoen and S. T. Yau. {\sl The energy and the linear momentum of space-times in general relativity}. Comm. Math. Phys. 79.1 (1981), pp. 47--51. DOI: \href{https://doi.org/10.1007/BF01208285}{10.1007/BF01208285}.


\bibitem[SY81b]{SY81b} R. Schoen and S. T. Yau. {\sl Proof of the positive mass theorem. II}. Comm. Math. Phys. 79.2 (1981), pp. 231–-260. DOI: \href{https://doi.org/10.1007/BF01942062}{10.1007/BF01942062}.


\bibitem[SY88]{SY88} R. Schoen and S. T. Yau. {\sl  Conformally flat manifolds, Kleinian groups and scalar curvature}. Invent. Math. 92.1 (1988), pp. 47--71. DOI: \href{https://doi.org/10.1007/BF01393992}{10.1007/BF01393992}.


\bibitem[SY17]{SY17} R. Schoen and S. T. Yau. {\sl  Positive scalar curvature and minimal hypersurface singularities}. Preprint. 2017. \href{https://arxiv.org/abs/1704.05490}{arXiv:1704.05490 [math.DG]}.


\bibitem[Wit81]{Wit81} E. Witten.  {\sl A new proof of the positive energy theorem}. Comm. Math. Phys. 80.3 (1981), pp. 381--402. DOI: \href{https://doi.org/10.1007/BF01208277}{10.1007/BF01208277}.


\bibitem[Zei20]{Zei20} R. Zeidler. {\sl Width, largeness and index theory}. SIGMA Symmetry, Integrability and Geometry: Methods and Applications, 16 (2020), Paper No. 127, 15. DOI: \href{https://doi.org/10.3842/SIGMA.2020.127}{10.3842/SIGMA.2020.127}.


\bibitem[Zei22]{Zei22} R. Zeidler.  {\sl Band width estimates via the Dirac operator}. J. Differential Geom. 122.1 (2022), pp. 155--183. DOI: \href{https://doi.org/10.4310/jdg/1668186790}{10.4310/jdg/1668186790}.


\bibitem[Zhu22]{Zhu22} J. Zhu. {\sl Positive mass theorem with arbitrary ends and its applications}. Preprint. 2022. \href{https://arxiv.org/abs/2204.05491}{arXiv:2204.05491 [math.DG]}. 

\end{thebibliography}

\newcommand{\etalchar}[1]{$^{#1}$}

\end{document}